\documentclass[a4paper]{article}

\pdfoutput=1

\usepackage[T1]{fontenc}
\usepackage[utf8]{inputenc}
\usepackage[english]{babel}

\usepackage[margin=29mm]{geometry}

\usepackage{amsmath, amssymb, amsthm}
\usepackage{mathtools}
\usepackage{bm}
\usepackage{physics}

\usepackage[affil-it]{authblk}
\usepackage[font=footnotesize]{caption}
\usepackage{newfloat}
\usepackage{subcaption}
\usepackage{booktabs}
\usepackage{graphicx}
\usepackage{adjustbox}

\usepackage{csquotes}
\usepackage[style=phys,biblabel=brackets,chaptertitle=false,eprint,url]{biblatex}

\usepackage{tikz}

\usepackage[colorlinks]{hyperref}
\usepackage[capitalise]{cleveref}

\usetikzlibrary{quantikz2}

\newtheorem{theorem}{Theorem}

\newtheorem{lemma}[theorem]{Lemma}

\newtheorem{corollary}[theorem]{Corollary}

\theoremstyle{definition}
\newtheorem{definition}[theorem]{Definition}

\theoremstyle{remark}
\newtheorem*{remark}{Remark}

\newcommand{\bigO}{O}
\newcommand{\softO}{\tilde{O}}
\DeclareMathOperator{\Cov}{Cov}
\DeclareMathOperator{\Var}{Var}
\DeclareMathOperator{\Past}{Past}
\DeclareMathOperator{\IPast}{IPast}
\DeclareMathOperator{\Fut}{Fut}
\DeclareMathOperator{\IFut}{IFut}

\DeclareFloatingEnvironment{algorithm}
\Crefname{algorithm}{\algorithmname}{\algorithmname s}

\title{Benincasa--Dowker--Glaser causal set actions \\ by quantum counting}

\author[ \hspace{-1ex}]{Sean A. Adamson\thanks{\href{mailto:sean.adamson@ed.ac.uk}{\texttt{sean.adamson@ed.ac.uk}}}}
\author[ \hspace{-1ex}]{Petros Wallden\thanks{\href{mailto:petros.wallden@ed.ac.uk}{\texttt{petros.wallden@ed.ac.uk}}}}
\affil[ \hspace{-1ex}]{School of Informatics, University of Edinburgh, \protect\\ 10 Crichton Street, Edinburgh EH8 9AB, United Kingdom}

\date{}

\begin{document}

\maketitle

\begin{abstract}
Causal set theory is an approach to quantum gravity in which spacetime is fundamentally discrete while retaining local Lorentz invariance. The Benincasa–Dowker–Glaser action is the causal set equivalent to the Einstein–Hilbert action underpinning Einstein's general theory of relativity. We present a $\tilde{O}(n^{2})$ running-time quantum algorithm to compute the Benincasa–Dowker–Glaser action in arbitrary spacetime dimensions for causal sets with $n$ elements which is asymptotically optimal and offers a polynomial speedup compared to all known classical or quantum algorithms. To do this, we prepare a uniform superposition over an $O(n^{2})$-size arbitrary subset of computational basis states encoding the classical description of a causal set of interest. We then construct depth $\tilde{O}(n)$ oracle circuits testing for different discrete volumes between pairs of causal set elements. Repeatedly performing a two-stage variant of quantum counting using these oracles yields the desired algorithm.
\end{abstract}

\section{Introduction}

Einstein's general theory of relativity, which describes gravity as a geometric property of spacetime, has proven to be one of the most successful physical theories for describing our universe at macroscopic scales \cite{clemence1947relativity,fomalont2009progress,abbott2016observation}.
However, its incompatibility with our best theory of microscopic scales, quantum mechanics, is well known.
It is therefore important to find a theory that is capable of describing scenarios where neither quantum nor relativistic effects may be ignored (such as near black holes or in the early universe).
The theory of \emph{causal sets} originally proposed by \textcite{bombelli1987space}, which considers the causal order of spacetime to be discrete and fundamental, is a candidate theory of quantum gravity on which such a description could possibly be built.
One of the most striking phenomenological aspects of causal sets was the correct prediction by \textcite{sorkin1990spacetime,sorkin1997forks} of the order of magnitude of the cosmological constant, which was later experimentally observed.

In general, an action has dual use---giving rise to the behavior of a system and as a means to quantize a theory by a path integral (sum over histories) approach.
The (classical) equations of motion for a physical system can be obtained from an action principle: The path taken by a physical system through its state space can be explained as being a stationary value of the system's action functional.
In general relativity, this role is played by the Einstein--Hilbert action.
The celebrated Einstein field equations \cite{einstein1916grundlage} were formulated by David Hilbert using the principle of least action applied to this functional \cite{hilbert1915grundlagen}.
In quantum mechanical theories, every possible path through state space contributes to the dynamics of the system.
The amplitude of each path (or ``history'') has its phase determined by the action.
The Einstein--Hilbert action is also used in approaches to quantize gravity \cite{hartle1983wave}.

It is therefore natural to seek a discrete action on causal sets whose expectation converges to the Einstein--Hilbert action in the continuum limit.
The main candidate for such a discrete action is the Benincasa--Dowker--Glaser (BDG) action \cite{benincasa2010scalar}, for which there is both analytical and numerical evidence that it converges in expectation up to some boundary terms \cite{benincasa2010scalar,surya2012evidence,benincasa2013action,dowker2021boundary,dowker2013causal}.
Causal sets are the only mathematical structures that are discrete and Lorentz invariant \cite{dowker2004quantum,dowker2021recovering}.
Therefore, the BDG action can be of phenomenological use for any fundamental theory that at some scale appears to satisfy these properties.
Like the path integral (sum over histories) formulation of quantum mechanics, causal set actions can be used to construct discrete quantum dynamics by acting as weights in sums over the sample space of causal sets in a ``sum over causal sets'' approach.
The development of the correct dynamics for causal sets is ongoing \cite{loomis2017suppression}.

\subsection{Our contributions}

Our main contribution is a quantum algorithm to estimate the value of the $d$-dimensional Benincasa--Dowker--Glaser causal set action up to an $\bigO(\varepsilon n)$ error in $\softO(n^{2} \varepsilon^{-1})$ time.
This is an improvement over classical algorithms that can be used, which have time complexities $\bigO(n^{3})$ or $\bigO(n^{2.8074})$ (see \cref{sec:matrix_multiplication} for the latter) depending on available memory ($\softO(1)$ for the former and $\softO(n^{2})$ for the latter) and acceptable constant factors in running time.
The quantum algorithm is described in \cref{alg:bd_quantum_algorithm}.
The statement of our main result found in \cref{thm:bd_quantum_algorithm} formulates the accuracy and resource requirements more precisely.

\begin{algorithm}[htb]
    \caption{
        Estimating the Benincasa--Dowker--Glaser causal set action for $d$ spacetime dimensions.
    }
    \label{alg:bd_quantum_algorithm}
    \rule{\linewidth}{0.08em}

    \begin{enumerate}
        \item For each $k \in \{0, \dots, n_{d} - 1\}$, where $n_{d} = \frac{d}{2} + 2$ for even dimension and $n_{d} = \frac{d-1}{2} + 2$ for odd dimension, perform the following steps to estimate order interval abundance $N_{k}$.
        \begin{enumerate}
            \item \textbf{Initial superposition.}
            Prepare an initial $\lceil 2 \log_{2}{n} \rceil$-qubit uniform superposition over the first $n^{2}$ basis states
            \begin{equation}
                \frac{1}{n} \sum_{j=0}^{n^2 - 1} \ket{j} .
            \end{equation}
            Achievable in $\bigO(\log{n})$ time using an efficient circuit (e.g. \textcite{shukla2024efficient}).
    
            \item \textbf{Data superposition.}
            Apply a sequence of $X_{t}^{c}$ gates as described in \cref{sec:data_prep} to produce the state with $\lceil 2 \log_{2}{n} \rceil + 2n$ qubits
            \begin{equation}
                \frac{1}{n} \sum_{i,j=1}^{n} \ket{h(i,j)} \otimes \ket{\bm{r}_{i} \bm{c}_{j}} ,
            \end{equation}
            where the $\bm{r}_{i}$ and $\bm{c}_{j}$ are the $n$-bit binary strings forming rows and columns of the $n \times n$ reflexive adjacency matrix of the causal set.
            Achievable in $\bigO(n^{2} \log{n})$ time (\cref{cor:uniform_superposition}).
    
            \item \textbf{Abundance counting.}
            Perform the two-stage quantum counting variant of \cref{sec:abundance_counting_alg} using the depth $\softO(n)$ oracle circuit $V_{k}$ given in \cref{fig:bd_oracle} to estimate $N_{k}$ by $\hat{N}_{k}$.
            Achievable in $\softO(n^{2})$ time (since $\bigO(n)$ oracle queries are made).
        \end{enumerate}

        \item Input all the $\hat{N}_{k}$ into the action formula of \cref{def:bd_action_d} to obtain an estimate $\hat{S}^{(d)}$ for $S^{(d)}$.
    \end{enumerate}

    \rule{\linewidth}{0.08em}
\end{algorithm}

We also analyze (\cref{sec:random_sampling}) a classical method for approximately counting order-interval abundances (and thus computing the BDG action) based on simple random sampling of pairs of elements.
In order to obtain the correct error scaling in this approximation, we find that the number of pairs sampled must be at least a constant multiple of the total number of adjacent elements in the causal set, which is $\bigO(n^{2})$ in general.
Although this does not offer an asymptotic improvement in the running time over the naive $\bigO(n^{3})$ time classical method, there is an improved constant factor that comes at the cost of its output being an approximation.

\subsection{Overview of techniques}

The basis of our algorithm is to use quantum counting \cite{brassard1998quantum,brassard2002quantum} to efficiently count the abundance of order intervals of a particular size in a causal set.
This counting is performed only approximately, so it is important to ensure that as the number of causal set elements $n$ being considered grows, the error in the approximation increases at most linearly in $n$.
To this end, in \cref{sec:abundance_counting_alg} we adapt the simplified quantum counting algorithm of \textcite{aaronson2020quantum} (which makes no use of the quantum Fourier transform and relies solely on Grover iterations) to perform two stages: an initial rough estimate of the count which is used to feedback and tune the target error parameter, followed by a final count using this updated parameter.
This results in an estimate with an error that is at most proportional to the square root of the true count (\cref{thm:approx_counting_sqrt_err}).

In order to apply the counting algorithm successfully, we require two other components: Efficient data preparation and an efficient quantum oracle.
For our algorithm, data preparation means the preparation of a uniform superposition state over binary strings of qubits representing all row and column pairs of the adjacency matrix for the (reflexive) causal set.
To achieve this, we first prepare a uniform superposition over the correct number $\bigO(n^{2})$ of states (corresponding to elements of an $n \times n$ matrix) inside a minimal-dimensional quantum register using the recent algorithm of \textcite{shukla2024efficient}.
We apply a circuit that uses the elements of this superposition to enumerate states in a second, larger register; mapping each to strings of qubits representing the row/column data.
For this, we assume that the classical description of the causal set is provided in a form that, for any given element, allows constant-time lookup of the set of all elements to its future and to its past.

We construct a family of quantum oracle circuits (parametrized by an integer $k \geq 0$) that, when acting on classical input states representing the adjacency data for some pair of causal set elements, determines whether or not the volume of their inclusive order interval takes a particular given value ($k+2$), which we may choose.
As this is a quantum circuit, it can be applied to the uniform superposition containing all of the causal set adjacency data for pairs of elements that we already prepared.
As in many quantum algorithms, exploiting this superposition ``parallelism'' forms the origin of the quantum speedup exhibited by the counting algorithm.
Having an algorithm counting the abundance of order intervals of any given particular volume is enough, as we only need to repeat this for a fixed number of different $k$ values to find the BDG action (with the number depending on the spacetime dimension being considered).
That the oracle circuits are efficient enough for our purposes follows from recent developments in the depths required to compile general multi-controlled Toffoli gates into single-qubit and controlled-\textsc{NOT} quantum gates with minimal space overhead \cite{claudon2024polylogarithmic}.
Any further advancements in this area could immediately be applied to improve the resource complexities of our algorithm.

\subsection{Related works}

The first explicit quantization of the causal structure of spacetime was arguably the ``causal spaces'' defined axiomatically by \textcite{kronheimer1967structure}.
Causal set theory was proposed by \textcite{bombelli1987space} based on a theorem of \textcite{hawking1976new} and its generalization by \textcite{malament1977class}.
The discrete BDG action that (under reasonable assumptions) has as its continuum limit the Einstein--Hilbert action was developed by \textcite{benincasa2010scalar} (for two- and four-dimensional spacetimes) and generalized to arbitrary dimensions by \textcite{dowker2013causal}.
Numerical simulations and heuristic arguments suggest the validity of the BDG action as the Einstein--Hilbert action in the continuum limit up to some boundary terms \cite{benincasa2013action,dowker2021boundary,dowker2013causal}.
Closed-form expressions for the action in arbitrary dimensions were derived by \textcite{glaser2014closed}.
Recent work of \textcite{yeats2025combinatorial} has shown that the coefficients can also be obtained combinatorially via chord diagrams.
\textcite{glaser2013towards} derived analytic expressions for the expectation values of order-interval abundances (which are quantities used to define the BDG action) in flat spacetime.
Optimizations adapted to classical numerical techniques to compute the BDG action have been explored \cite{cunningham2018high}.
Other attempts have also been made to define the action of a causal set \cite{sverdlov2009gravity,roy2013discrete}.
For further details on causal set theory, see the excellent review article by \textcite{surya2019causal}.

A quantum algorithm for the unstructured search problem (identifying the unique marked item from a black-box function that produces a particular output value) was given by \textcite{grover1996fast}.
This seminal algorithm offers a quadratic improvement in query complexity over any possible classical method and is provably optimal in the circuit model of quantum computing \cite{bennett1997strengths}.
The idea was generalized by \textcite{brassard1998quantum,brassard2002quantum} to the scenario in which there are multiple marked items and was also used to count them, showing the same quantum speedup.
A simplified version of this \emph{quantum counting} that does not make use of the quantum Fourier transform (and relies only on Grover iterations) was given by \textcite{aaronson2020quantum}.
\textcite{shukla2024efficient} showed an efficient algorithm to prepare a uniform superposition of an arbitrary number of states, improving previous methods \cite{babbush2018encoding}.
We make use of fan-out gates to expand the number of target qubits of a controlled gate, and note that bounded fan-out gates are almost within reach of quantum hardware \cite{rasmussen2020single,guo2022implementing,kim2022high,fenner2023implementing}.
Quantum circuits for compiling multi-controlled generalized Toffoli gates in polylogarithmic depth using only a single ancilla qubit were found by \textcite{claudon2024polylogarithmic}.
\textcite{khattar2025rise} introduced ``conditionally clean'' ancilla qubits, which can be used to reduce the space overhead of circuits including generalized Toffoli and related gates such as incrementers and comparators.
Although no quantum algorithm for general $n \times n$ matrix multiplication beating the current classical time complexity record of $\bigO(n^{2.371339})$ \cite{alman2025more} is known, in the special cases over the Boolean semiring with sparse input or output matrices polynomial speedups are known \cite{jeffery2012improving,legall2014quantum}.
On the other hand, matrix products can be classically verified in $\bigO(n^{2})$ time \cite{freivalds1979fast}.
In the quantum case, there is an $\bigO(n^{5/3})$ time quantum algorithm by \textcite{buhrman2006quantum}.
For further details about quantum computing and quantum information, see, for example, \cite{nielsen2010quantum}.

\subsection{Organization of the paper}

In \cref{sec:preliminaries} we introduce causal sets, describe the Benincasa--Dowker--Glaser causal set action, and give some preliminary results on quantum circuits and the quantum counting algorithm that we later make use of.
In \cref{sec:problem} we describe the problem of algorithmically computing the BDG action and introduce a possible approach to this based on the adjacency matrix of the causal set.
In \cref{sec:classical_approaches} we discuss two exact algorithms for the problem (the basic naive approach and a classical matrix approach), followed by a brief analysis of an approximate classical method that randomly samples pairs of causal set elements.
\Cref{sec:quantum_approach} details our quantum approach based on quantum counting: We show how to prepare the correct superposition state, design a quantum oracle circuit suited to our purpose, and exhibit a variant of quantum counting suitable for our purposes.
We analyze the relevant time and space complexities throughout before combining these components with quantum counting for our final algorithm.
In \cref{sec:discussion} we conclude with a discussion of our work and outline some possible future directions.

\section{Preliminaries}
\label{sec:preliminaries}

We introduce causal sets and the BDG action in \cref{sec:causal_sets}.
In \cref{sec:causal_graphs_matrices}, we explain the representation of causal sets using graphs and matrices.
In \cref{sec:uniform_superposition,sec:mcx}, we give results on generating quantum superposition states and on quantum gates relevant to our algorithm.
The quantum counting algorithm is stated in \cref{sec:quantum_counting}.

\subsection{Causal sets}
\label{sec:causal_sets}

A \emph{causal set} (or \emph{causet} for short) is a locally finite partially ordered set.

\begin{definition}[Causal set]
\label{def:causet}
    A \emph{causal set} is a set $C$ on which a homogeneous relation $\prec$ is defined that satisfies the following properties.
    \begin{enumerate}
        \item Irreflexivity: For all $x \in C$, we have not $x \prec x$.
        \item Asymmetry: For all $x, y \in C$, if $x \prec y$ then not $y \prec x$.
        \item Transitivity: For all $x, y, z \in C$, if $x \prec y$ and $y \prec z$ then $x \prec z$.
        \item Local finiteness: For all $x, z \in C$, we have that $\{ y \in C \mid x \prec y \text{ and } y \prec z \}$ is a finite set.
    \end{enumerate}
\end{definition}

\begin{remark}
    We have chosen to use the irreflexive (strict) definition of partial orders, but we could equivalently use the reflexive definition.
    We will denote by $\preceq$ the reflexive closure of $\prec$ given by
    \begin{equation}
        x \preceq y \text{ if } x \prec y \text{ or } x = y .
    \end{equation}
    We will also denote by $\succ$ and $\succeq$ the duals of $\prec$ and $\preceq$, respectively.
\end{remark}

\begin{definition}[Past and future sets]
    For all $x$ in a causal set $C$ we define the following sets of points in $C$.
    \begin{enumerate}
        \item The \emph{exclusive past} (or simply \emph{past}) of $x$ is $\Past(x) = \{ e \in C \mid e \prec x \}$.

        \item The \emph{exclusive future} (or simply \emph{future}) of $x$ is $\Fut(x) = \{ e \in C \mid x \prec e \}$.
    \end{enumerate}
    We also define \emph{inclusive} versions of these, which contain the point $x$ itself.
    \begin{enumerate}
        \item The \emph{inclusive past} of $x$ is $\IPast(x) = \{ e \in C \mid e \preceq x \}$.

        \item The \emph{inclusive future} of $x$ is $\IFut(x) = \{ e \in C \mid x \preceq e \}$.
    \end{enumerate}
\end{definition}

\begin{definition}[Order intervals]
    For any two points $x$ and $y$ in a causal set $C$, their (\emph{exclusive}) \emph{order interval} $I(x,y)$ is the set of points
    \begin{equation}
        I(x,y) = \Fut(x) \cap \Past(y) .
    \end{equation}
    Their \emph{inclusive} order interval is denoted by
    \begin{equation}
        I[x,y] = \IFut(x) \cap \IPast(y) .
    \end{equation}
    Similarly, we define \emph{half-open} order intervals
    \begin{subequations}
    \begin{align}
        I[x,y) & = \IFut(x) \cap \Past(y) , \\
        I(x,y] & = \Fut(x) \cap \IPast(y) .
    \end{align}
    \end{subequations}
    We may refer to the cardinality of an order interval $n(x,y) = \lvert I(x,y) \rvert$ as its \emph{discrete volume}, with a similar notation defined for the other types of intervals.
\end{definition}

\begin{definition}[Nearest neighbors]
    For any points $x$ and $y$ in a causal set $C$ satisfying $x \prec y$, we say that $x$ and $y$ are \emph{$k$-nearest neighbors} if and only if $\lvert I(x,y) \rvert = k$.
    We also define the set of all $k$-nearest neighbors to the past of $x$ by
    \begin{equation}
        L_{k}(x) = \{ e \in C \mid e \prec x \text{ and } \lvert I(e,x) \rvert = k \} .
    \end{equation}
\end{definition}
\begin{remark}
    The $0$-nearest neighbors are simply the \emph{nearest} neighbors, the $1$-nearest neighbors are the \emph{next nearest} neighbors, etc.
\end{remark}

We now define the Benincasa--Dowker--Glaser (BDG) causal set action \cite{benincasa2010scalar} for four-dimensional spacetimes and a discreteness length $l$.
First, we write a discrete version of the scalar curvature.

\begin{definition}[Scalar curvature, four dimensions]
    The discrete scalar curvature at a point $x$ in a causal set $C$ in four dimensions is defined as
    \begin{equation}
        R(x) = \frac{4}{\sqrt{6}} \frac{1}{l^{2}} [1 - N_{0}(x) + 9 N_{1}(x) - 16 N_{2}(x) + 8 N_{3}(x)] ,
    \end{equation}
    where $N_{k}(x) = \lvert L_{k}(x) \rvert$.
\end{definition}

Next, we can sum over all points to obtain a discrete action.

\begin{definition}[Benincasa--Dowker--Glaser action, four dimensions]
\label{def:bd_action_4}
    The Benincasa--Dowker--Glaser action of a finite causal set $C$ with $n$ elements in four dimensions is given by
    \begin{equation}
    \begin{split}
        \frac{1}{\hbar} S(C)
        & = \frac{l^{4}}{l_{p}^{2}} \sum_{x \in C} R(x) \\
        & = \frac{4}{\sqrt{6}} \mathopen{}\left( \frac{l}{l_{p}} \right)^{2} [n - N_{0} + 9 N_{1} - 16 N_{2} + 8 N_{3}] ,
    \end{split}
    \end{equation}
    where $N_{k} = \sum_{x \in C} N_{k}(x)$ is the total number of $k$-nearest neighbor pairs in $C$.
\end{definition}

Similarly, the action can also be generalized to $d$ dimensions.

\begin{definition}[Benincasa--Dowker--Glaser action, $d$ dimensions]
\label{def:bd_action_d}
    The Benincasa--Dowker--Glaser action of a finite causal set $C$ with $n$ elements in $d$ dimensions is given by
    \begin{equation}
        \frac{1}{\hbar} S^{(d)}(C) = - \alpha_{d} \mathopen{}\left( \frac{l}{l_{p}} \right)^{d-2} \mathopen{}\left[ n + \frac{\beta_{d}}{\alpha_{d}} \sum_{k=0}^{n_{d} - 1} C_{k+1}^{(d)} N_{k} \right]\mathclose{} ,
    \end{equation}
    where $N_{k}$ is the total number of $k$-nearest neighbor pairs in $C$, the number of layers to be summed over is $n_{d} = \frac{d}{2} + 2$ for even dimensions and $n_{d} = \frac{d-1}{2} + 2$ for odd dimensions, and $\alpha_{d}$, $\beta_{d}$, and $C_{k}^{(d)}$ have closed-form expressions \cite{dowker2013causal,glaser2014closed}.
\end{definition}

In \cref{def:bd_action_4,def:bd_action_d} for the BDG action, the quantities $N_{k}$ denote the total numbers of $k$-nearest neighbor pairs in the causal set $C$ under consideration.
In the case of $N_{0}$, this is not quite the same as the number of order intervals such that $\lvert I(x,y) \rvert = 0$ since this is also satisfied when $x \nprec y$.
In order to distinguish between the cases of nearest neighbors and $x \nprec y$ in terms of discrete volumes, it is natural to instead use \emph{inclusive} order intervals.
We now rewrite the $k$-nearest neighbor \emph{abundances} in terms of order intervals.

\begin{definition}[Abundances]
\label{def:abundances}
    For a causal set $C$, the \emph{$k$-abundance} of $k$-nearest neighbor pairs is
    \begin{equation}
    \label{eq:abundance}
        N_{k} = \lvert \{ (x,y) \in C^{2} \mid n[x,y] = k + 2 \} \rvert ,
    \end{equation}
    where $n[x,y] = \lvert I[x,y] \rvert$ is the inclusive discrete volume between points $x$ and $y$.
\end{definition}

\subsection{Causal sets as graphs and matrices}
\label{sec:causal_graphs_matrices}

In the following, we restrict our discussion to finite causal sets, as it is for these that the BDG action is defined.
Thus, local finiteness is automatically satisfied, and we simply deal with finite partially ordered sets.

Any irreflexive partially ordered set $(C, \prec)$ can equivalently be expressed as a directed acyclic graph (DAG) denoted $G = (V,E)$.
The graph is constructed by taking each element of $C$ to be a vertex and each element of $\prec$ to be an edge ($V = C$ and $E = {\prec}$).
DAGs constructed in this way are automatically transitively closed and have the initial partial order as their reachability relation.
That these graphs are acyclic is equivalent to the asymmetry property of partial orders (or, in the language of causal sets, that there are no closed causal curves).
Conversely, any DAG (such as those represented by Hasse diagrams) has a reachability relation that is a strict partial order and that can be found directly from the transitive closure of the DAG.

The directed graph associated with a \emph{reflexive} partially ordered set is not acyclic, since it has loops connecting every vertex to itself.
It is for this reason that we chose to define causal sets using the irreflexive convention in \cref{def:causet}.

It is convenient to enumerate the elements of a causal set (or the vertices of its graph) by means of labeling each by a natural number.
For a partially ordered set $(C, \prec)$ of $n$ elements, if this enumeration is given by the bijection $f \colon \{ 1, \dots, n \} \to C$, we write $i \prec j$ to mean $f(i) \prec f(j)$.
Since the graph for a partially ordered set is acyclic, it is moreover always possible to enumerate its elements such that the usual total order on the index lower set of natural numbers describes a topological ordering of the graph.
In the language of order theory: There exists a linear extension to the partial order.
We use the term \emph{topological ordering} to refer to enumerations representing total orders of the vertices of directed acyclic graphs, while the (more general) term \emph{linear extension} is reserved for compatible total orders when considering partial orders themselves (which, unlike those for graphs, may not be finite in general).

\begin{theorem}[Topological ordering]
    For any directed acyclic graph $G = (V,E)$, there exists a topological ordering of $G$.
    That is, there exists an enumeration $f \colon \{ 1, \dots, \lvert V \rvert \} \to V$ such that for every edge $(x,y) \in E$ we have $f^{-1}(x) < f^{-1}(y)$.
\end{theorem}

\begin{remark}
    Standard algorithms for finding topological orderings (Kahn's algorithm \cite{kahn1962topological} and those based on depth-first search \cite{tarjan1976edge}) have time complexities $\bigO(\lvert V \rvert + \lvert E \rvert)$ and space complexities $O(\lvert V \rvert)$.
\end{remark}

We may now assume that the elements of our sets are enumerated (although it will be unnecessary to assume a topological ordering unless otherwise stated).
This makes available to us a representation of causal sets as matrices.

\begin{definition}[Adjacency matrix, directed graphs]
    The \emph{adjacency matrix} $\bm A$ of a directed graph $G = (V,E)$ with $n = \lvert V \rvert$ vertices is the $n \times n$ logical matrix defined by
    \begin{equation}
        A_{ij} =
        \begin{cases}
            1 & \text{if $(i,j) \in E$,} \\
            0 & \text{otherwise.}
        \end{cases}
    \end{equation}
\end{definition}

\begin{definition}[Adjacency matrix, partially ordered sets]
    The \emph{adjacency matrix} $\bm A$ of a finite partially ordered set $(C, \prec)$ with $n$ elements is the adjacency matrix of its corresponding transitively closed directed graph.
    Concretely, it is the $n \times n$ logical matrix defined by
    \begin{equation}
        A_{ij} =
        \begin{cases}
            1 & \text{if $i \prec j$,} \\
            0 & \text{otherwise.}
        \end{cases}
    \end{equation}
\end{definition}

\begin{remark}
    Since $\prec$ is irreflexive, this matrix has diagonal elements $A_{ii} = 0$.
    Moreover, if $C$ is enumerated in topological order then the matrix is upper triangular.
\end{remark}

\subsection{Quantum computation}
\label{sec:quantum_computation}

In quantum computation, qubits (quantum bits) are two-dimensional quantum systems acting analogously to the binary bits of classical computation.
The choice of two orthonormal states, denoted by $\ket{0}$ and $\ket{1}$, together forms the standard basis (called the \emph{computational} basis) of a qubit system $\mathbb{C}^{2}$.
Qubits may also be represented by complex matrices, and in this representation the basis states are such that
\begin{equation}
    \ket{0} \equiv
    \begin{pmatrix}
        1 \\ 0
    \end{pmatrix} ,\quad
    \ket{1} \equiv
    \begin{pmatrix}
        0 \\ 1
    \end{pmatrix} .
\end{equation}
Many qubits taken together (as a tensor product) form what is known as a quantum \emph{register}, whose standard basis states (also called computational) are typically denoted in the form $\ket{x}$, where the $x$ are binary strings (and sometimes denoted by their corresponding digits).
For example, the computational basis of a two-qubit system is written as $\{ \ket{00}, \ket{01}, \ket{10}, \ket{11} \}$.
The vectors comprising this basis can again be represented by matrices such that
\begin{equation}
    \ket{00} \equiv
    \begin{pmatrix}
        1 \\ 0 \\ 0 \\ 0
    \end{pmatrix} ,\quad
    \ket{01} \equiv
    \begin{pmatrix}
        0 \\ 1 \\ 0 \\ 0
    \end{pmatrix} ,\quad
    \ket{10} \equiv
    \begin{pmatrix}
        0 \\ 0 \\ 1 \\ 0
    \end{pmatrix} ,\quad
    \ket{11} \equiv
    \begin{pmatrix}
        0 \\ 0 \\ 0 \\ 1
    \end{pmatrix} .
\end{equation}
In general, the state of an $n$-qubit quantum register is represented by a unit vector in a $2^{n}$-dimensional complex Hilbert space.

A quantum computer is a device that is able to manipulate a quantum register by applying unitary operators (called quantum logic gates in this context) after it has been prepared in some initial state (conventionally taken to be $\ket{0 \dots 0}$).
Some frequently used quantum gates are the Hadamard gate $H$, the Pauli-$X$ gate, and the controlled \textsc{NOT} gate \textsc{CNOT}.
In terms of their matrix representations,
\begin{equation}
    H \equiv \frac{1}{\sqrt{2}}
    \begin{pmatrix}
        1 & 1 \\
        1 & -1
    \end{pmatrix} ,\quad
    X \equiv
    \begin{pmatrix}
        0 & 1 \\
        1 & 0
    \end{pmatrix} ,\quad
    \textsc{CNOT} =
    \begin{pmatrix}
        1 & 0 & 0 & 0 \\
        0 & 1 & 0 & 0 \\
        0 & 0 & 0 & 1 \\
        0 & 0 & 1 & 0
    \end{pmatrix} .
\end{equation}
It is typical to consider devices that may only utilize a certain finite set of gates, a finite sequence of which (called a \emph{quantum circuit}) can be made to arbitrarily approximate any $n$-qubit unitary operator.
Such a set of quantum gates is called \emph{universal} for quantum computation, and the approximation can be done remarkably efficiently \cite{kitaev1997quantum,dawson2006solovay}.
After applying the desired quantum circuit, the qubit register can be measured in the computational basis to obtain the final result of the quantum computation.
One example of a universal quantum gate set is the Hadamard gate together with the Toffoli gate \textsc{CCNOT} (also called the controlled controlled \textsc{NOT} gate),
\begin{equation}
    \textsc{CCNOT} =
    \begin{pmatrix}
        I_{6} & 0 \\
        0 & X
    \end{pmatrix} .
\end{equation}

The resources required to implement a quantum computation (with respect to some choice of quantum gate set) can be quantified according to a number of figures of merit.
\begin{itemize}
    \item
    Quantum circuit \emph{depth} quantifies the number of time steps from the beginning to the end of a quantum circuit implementing the desired computation, counting gates that can be executed in parallel as belonging to the same step.
    This quantity allows us to compare the amount of physical time required to perform different computations.

    \item
    Quantum circuit \emph{width} is the maximum number of qubits used at any given point between the beginning and end of a quantum circuit.
    This quantity corresponds to the amount of space (quantum register size) required for a computation.
\end{itemize}
Denoting by $n$ a quantity parametrizing a computational problem (in our case $n$ is the number of causal set elements), we often talk about the required computational resources in terms of their computational \emph{complexities}.
That is, we use big $\bigO$ notation to express a worst-case bound on how these resources scale with respect to $n$ as $n$ tends to infinity.
We interchangeably refer to the quantum circuit depth complexity as the time complexity and the quantum width complexity as the space complexity.
We refer to a time complexity that is $\bigO(1)$ as \emph{constant} time.

\subsection{Uniform superposition states}
\label{sec:uniform_superposition}

In order to encode our classical causal set data for quantum computation, we need to prepare an equal superposition of an arbitrary number of $m$-qubit computational basis states.
Due to an algorithm of \textcite{shukla2024efficient}, the first $M \leq 2^{m}$ such basis states can be efficiently prepared in equal superposition
\begin{equation}
    \frac{1}{\sqrt{M}} \sum_{j=0}^{M-1} \ket{j} .
\end{equation}
The circuit depth is $\bigO(\log{M})$ and only $\lceil \log_{2}{M} \rceil$ qubits (the smallest number to encode $M$ basis states possible) are required.

\subsection{Multi-controlled \textsc{NOT} gates}
\label{sec:mcx}

We make frequent use of $m$-controlled \textsc{NOT} gates $\textsc{C}^{m}\textsc{NOT}$.
Their action is to flip a single computational basis qubit if and only if all $m$ control qubits are excited
\begin{equation}
    \textsc{C}^{m}\textsc{NOT} = \dyad{1 \dots 1} \otimes \textsc{NOT} + (I - \dyad{1 \dots 1}) \otimes I .
\end{equation}
These $m$-controlled \textsc{NOT} gates can be exactly decomposed into primitive single-qubit and \textsc{CNOT} gates with polylogarithmic $\Theta(\log(m)^{3})$ circuit depth and a single (borrowed) ancilla qubit \cite{claudon2024polylogarithmic}.
Alternatively, if $m-2$ ancilla qubits are available rather than just one, then it is simple to implement $\textsc{C}^{m}\textsc{NOT}$ with depth $\bigO(\log{m})$.
One can use $m-1$ Toffoli gates to perform a cascaded \textsc{AND} operation starting on the $m$ control qubits, performing individual \textsc{AND}s on disjoint pairs of qubits and storing the final result in the target qubit, before uncomputing using another $m-2$ Toffoli gates.
This is illustrated in \cref{fig:simple_mcx} with $m=4$.

\begin{figure}[htb]
    \centering
    \begin{equation*}
        \begin{quantikz}
            \lstick[4]{$\mathcal{C}$} & \ctrl{4} & \\
            & \control{} & \\
            & \control{} & \\
            & \control{} & \\
            \lstick{$\mathcal{T}$} & \targ{} &
        \end{quantikz}
        =
        \begin{quantikz}
            \lstick[4]{$\mathcal{C}$} & \ctrl{4}\gategroup[6, steps=2, style={dashed, red, rounded corners}]{} & & & \gategroup[6, steps=2, style={dashed, red, rounded corners}]{} & \ctrl{4} & \\
            & \control{} & & & & \control{} & \\
            & & \ctrl{3} & & \ctrl{3} & & \\
            & & \control{} & & \control{} & & \\
            \lstick[2]{$\ket{0}$} & \targ{} & & \ctrl{2} & & \targ{} & \rstick[2]{$\ket{0}$} \\
            & & \targ{} & \control{} & \targ{} & & \\
            \lstick{$\mathcal{T}$} & & & \targ{} & & &
        \end{quantikz}
    \end{equation*}
    \caption{
        Decomposition of the $\textsc{C}^{4}\textsc{NOT}$ gate into Toffoli gates using two ancilla qubits prepared as $\ket{0}$.
        Control and target qubit registers are labeled $\mathcal{C}$ and $\mathcal{T}$, respectively.
        Gates grouped into dashed red boxes can be performed simultaneously.
        This grouping is the origin of the $\bigO(\log{m})$ depth for general $\textsc{C}^{m}\textsc{NOT}$ gates using this decomposition.
    }
    \label{fig:simple_mcx}
\end{figure}
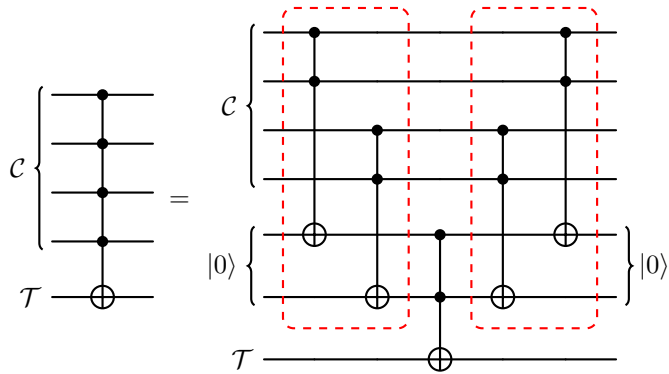

We also require a further generalization to $m$-controlled $n$-target \textsc{NOT} gates, which we denote by $\textsc{C}^{m}\textsc{NOT}^{n}$.
These gates are defined in the computational basis to flip all $n$ qubits in the target register if and only if all $m$ control qubits are excited.
A $\textsc{C}^{m}\textsc{NOT}^{n}$ gate can be implemented with two $\textsc{C}^{m}\textsc{NOT}$ gates and an $n$-target \textsc{CNOT} gate (also called fan-out gate) using an ancilla qubit, as shown in \cref{fig:mcmx}.
The circuit depth of this then depends on the implementation of the $\textsc{C}^{m}\textsc{NOT}$ gates used and the depth of fan-out.
An $n$-target fan-out gate can be compiled with depth $\Theta(\log{n})$ \cite{green2002counting,broadbent2009parallelizing}.
Therefore, with a total of two ancilla qubits available, the depth is $\bigO(\log(m)^{3} + \log{n})$; while with $m-1$ ancilla qubits available, the depth can be reduced to $\bigO(\log(mn))$.

\begin{figure}[htb]
    \centering
    \begin{equation*}
        \begin{quantikz}
            \lstick{$\mathcal{C}$} & \qwbundle{m} & \ctrl{1} & \\
            \lstick{$\mathcal{T}$} & \qwbundle{n} & \targ{} &
        \end{quantikz}
        =
        \begin{quantikz}
            \lstick{$\mathcal{C}$} & \qwbundle{m} & \ctrl{1} & & \ctrl{1} & \\
            \lstick{$\ket{0}$} & & \targ{} & \ctrl{1} & \targ{} & \rstick{$\ket{0}$} \\
            \lstick{$\mathcal{T}$} & \qwbundle{n} & & \targ{} & &
        \end{quantikz}
    \end{equation*}
    \caption{
        The $\textsc{C}^{m}\textsc{NOT}^{n}$ $m$-controlled $n$-target \textsc{NOT} gate expressed in terms of $\textsc{C}^{m}\textsc{NOT}$ gates and an $n$-target \textsc{NOT} gate using one additional ancilla qubit prepared as $\ket{0}$.
        Control and target qubit registers are labeled $\mathcal{C}$ and $\mathcal{T}$, respectively.
    }
    \label{fig:mcmx}
\end{figure}
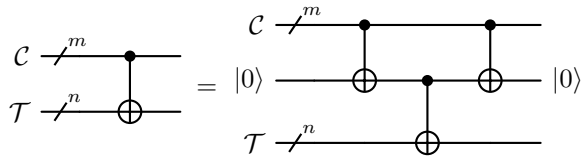

\subsection{Quantum oracles and quantum counting}
\label{sec:quantum_counting}

Consider some function $f \colon \{0, \dots, N-1\} \to \{0, 1\}$ whose role is to \emph{mark} $K = \lvert f^{-1}(1) \rvert > 0$ from $N$ items with a ``$1$'', and otherwise assign a ``$0$''.
A \emph{membership oracle} for $f$ refers to the unitary operator $U_{f}$ satisfying $U_{f} \ket{x} = (-1)^{f(x)} \ket{x}$.
That is, to any computational basis state $\ket{x}$, the membership oracle for $f$ applies a multiplicative phase of $-1$ if and only if $f(x) = 1$ (the state is marked).
There may be many different quantum circuit representations of a given membership oracle $U_{f}$, and we refer to any such circuit implementing $U_{f}$ as a \emph{quantum oracle circuit} for $f$.
A seminal algorithm due to \textcite{grover1996fast} achieves a quadratic improvement in time complexity over best possible classical case.
Moreover, this quantum speedup is asymptotically tight \cite{bennett1997strengths}.

In general, implementations of Grover's algorithm for unstructured search problems (starting with classical input data) take the form of the following steps.
\begin{enumerate}
    \item Compute a classical description of the function as a quantum oracle circuit.
    \item Prepare a uniform superposition of the $N$ particular states we want to search over.
    \item Run the Grover search algorithm:
    \begin{enumerate}
        \item Grover iterations: Repeatedly apply the oracle circuit and \emph{Grover diffusion} operator.
        \item Measure the output state in the computational basis.
    \end{enumerate}
\end{enumerate}
If it is known that there are $K$ matching states, then only $\bigO(\sqrt{N/K})$ Grover iterations are required to measure a correct answer with high probability.
However, it is also important for practical implementations that the time taken for classical computation of the appropriate quantum oracle, preparation of the correct superposition of input states, and each application of the quantum oracle circuit do not negate the quantum speedup gathered from needing only few Grover iterations.
If the number of solutions $K$ is unknown, we can employ the closely related \emph{quantum counting} algorithms to find it \cite{brassard1998quantum,brassard2002quantum,aaronson2020quantum}, again with the same concerns for practical implementations.

The following result of \textcite{aaronson2020quantum} regards a simplified approximate quantum counting algorithm that makes use of only Grover iterations (rather than relying on the quantum Fourier transform as in the original algorithm \cite{brassard1998quantum,brassard2002quantum}).
\begin{theorem}[Approximate counting \cite{aaronson2020quantum}]
\label{thm:approx_counting}
    Let $f \colon \{0, \dots, N-1\} \to \{0, 1\}$ be a Boolean function marking $K = \lvert f^{-1}(1) \rvert > 0$ items.
    Given access to a membership oracle for $f$ and $\delta, \epsilon > 0$, there exists a quantum algorithm that outputs an estimate $\hat{K}$ for $K$.
    The output $\hat{K}$ satisfies
    \begin{equation}
        \lvert \hat{K} - K \rvert < \epsilon K
    \end{equation}
    with probability at least $1 - \delta$.
    The algorithm makes fewer than
    \begin{equation}
        \bigO \mathopen{}\left( \epsilon^{-1} \sqrt{\frac{N}{K}} \log{\frac{1}{\delta}} \right)\mathclose{}
    \end{equation}
    oracle queries.
    The algorithm requires $\bigO(\log{N})$ qubits of space.
\end{theorem}
\begin{remark}
    Recall that a membership oracle for $f$ refers to the unitary operator $U_{f}$ satisfying $U_{f} \ket{x} = (-1)^{f(x)} \ket{x}$.
\end{remark}

If $\epsilon = \varepsilon K^{-1/2}$ for some $\varepsilon > 0$, then the estimate in \cref{thm:approx_counting} would satisfy $\lvert \hat{K} - K \rvert < \varepsilon \sqrt{K}$ and the algorithm would make fewer than $\bigO \mathopen{}\bigl( \varepsilon^{-1} \sqrt{N} \log(\delta^{-1}) \bigr)\mathclose{}$ oracle queries.
However, $K$ is unknown since it is the quantity that we are attempting to estimate, so we cannot choose $\epsilon$ this way.
We show how this can be achieved in \cref{sec:abundance_counting_alg}, with the precise statement given as \cref{thm:approx_counting_sqrt_err}.

\section{Framing the problem}
\label{sec:problem}

Given a classical description of a finite causal set of $n$ elements and a dimension of interest $d$, we would like to calculate its BDG action as expressed in \cref{def:bd_action_d}.
We assume that the constants and coefficients $n_{d}$, $\alpha_{d}$, $\beta_{d}$, and $C_{k}^{(d)}$ are already known (for a given $d$ of interest these can easily be computed beforehand and reused for future calculations in the same dimension).
We also assume that the causal set is provided in a form in which $x \prec y$ can be checked for any pair of points in constant time with respect to $n$.
It is enough to calculate each of the abundances $N_{k}$ (defined in \cref{def:abundances}) appearing in \cref{def:bd_action_d} separately, as the number of these quantities appearing in the expression is finite and independent of $n$ (it is a constant for a fixed dimension $d$).
All quantities can then be combined to find the BDG action.
The time and space complexities of the combined algorithm are then equal to those of calculating the $N_{k}$ sequentially for each $k$; that is, they are the highest order complexities of calculating individual $N_{k}$ over all $k$.
Our goal has thus been simplified to find all these $N_{k}$ efficiently for any given causal set.

\subsection{Matrix of volumes}
\label{sec:volume_matrix}

The $i$th row of the adjacency matrix $\bm A$ has its $k$th element marked by a $1$ if and only if $k$ lies in the future of $i$.
Similarly, the $j$th column of the matrix has its $k$th element marked by a $1$ if and only if $k$ lies in the past of $j$.
These are the statements
\begin{subequations}
\begin{align}
    \Fut(i) & = \{ k \mid A_{ik} = 1 \} , \\
    \Past(j) & = \{ k \mid A_{kj} = 1 \} .
\end{align}
\end{subequations}
Thus, we see that the exclusive discrete volumes can be written as the dot product of a row and a column
\begin{equation}
    \lvert I(i,j) \rvert
    = \lvert \{ k \mid A_{ik} = A_{kj} = 1 \} \rvert
    = \sum_{k} A_{ik} A_{kj} .
\end{equation}
The matrix of all exclusive discrete volumes $\bm V$ defined by $V_{ij} = \lvert I(i,j) \rvert$ is therefore simply
\begin{equation}
    \bm V = \bm A^{2} .
\end{equation}

So that we can easily distinguish between entries of $V_{ij}$ that correspond to $0$-nearest neighbors and those such that $i \nprec j$, we would instead like to compute the matrix of \emph{inclusive} discrete volumes $\bm{\bar{V}}$ defined by $\bar{V}_{ij} = \lvert I[i,j] \rvert$.
This can be achieved by instead considering the adjacency matrix of the graph for the reflexive relation $\preceq$ (where every element is also related to itself) given by $\bm{\bar{A}} = \bm A + \bm I$.
We will call this matrix $\bm{\bar{A}}$ the \emph{reflexive} adjacency matrix.
Using this, we have that
\begin{subequations}
\begin{align}
    \IFut(i) & = \{ k \mid \bar{A}_{ik} = 1 \} , \\
    \IPast(j) & = \{ k \mid \bar{A}_{kj} = 1 \} .
\end{align}
\end{subequations}
The inclusive discrete volumes can be written as
\begin{equation}
    \lvert I[i,j] \rvert
    = \lvert \{ k \mid \bar{A}_{ik} = \bar{A}_{kj} = 1 \} \rvert
    = \sum_{k} \bar{A}_{ik} \bar{A}_{kj} .
\end{equation}
Therefore, the matrix of inclusive discrete volumes can be expressed as
\begin{equation}
    \bm{\bar{V}}
    = \bm{\bar{A}}^{2}
    = (\bm A + \bm I)^{2} .
\end{equation}
Unlike the adjacency matrix $\bm A$, the matrix $\bm{\bar{A}}$ has the useful property that it is invertible.
Moreover, if $(C, \prec)$ is enumerated in topological order then $\bm{\bar{A}}$ is upper unitriangular.

\begin{remark}
    An alternative method to identify $0$-nearest neighbor pairs of elements $i \prec j$ without using $\bm{\bar{A}}$ is to compute $\bm{V} = \bm{A}^{2}$ and, where $V_{ij} = 0$, check that $A_{ij} = 1$.
\end{remark}

Finally, the $k$-abundances $N_{k}$ can be found by counting the number of times the value $k+2$ appears in the inclusive volume matrix $\bm{\bar{V}}$ for each $k$ desired ($0 \leq k \leq 3$ for computing the four-dimensional BDG action of \cref{def:bd_action_4}).
Note that the $k$-abundances do not depend on the positions of the values $k+2$ in the matrix, but only on their count.
Therefore, no topological ordering of the input data need be assumed.

\section{Classical approaches}
\label{sec:classical_approaches}

Let us first consider possible approaches to numerically evaluating the BDG action of a causal set with $n$ elements, where we are limited to classical computational devices.
We will also denote the number of causally related pairs of points (edges in the corresponding DAG) by $N = \lvert E \rvert = \bigO(n^{2})$.

\subsection{Exact methods}
\label{sec:exact_methods}

Naively evaluating the action involves computing $\bigO(n^{2})$ discrete volumes between the $N$ causally related pairs of elements (each of which takes $\bigO(n)$ time) and adding to the count of each particular possible volume of interest each time.
The overall time complexity is thus $\bigO(n^{3})$.
The space complexity is $\softO(1)$.
Classical optimizations (which do not affect the asymptotic complexities) for this method have been previously explored \cite{cunningham2018high}.

\subsubsection{Matrix methods}
\label{sec:matrix_multiplication}

As discussed in \cref{sec:problem}, finding the abundances of discrete volumes is equivalent to counting the occurrences of different values in $(\bm A + \bm I)^{2}$.
Thus, any sub-cubic time algorithm for computing the matrix square yields a sub-cubic time algorithm for the action.
There indeed exist $\bigO(n^{\omega})$ time classical matrix multiplication algorithms, where the minimal value of $2 \leq \omega < 3$ is an open question in theoretical computer science.

The fastest presently known matrix multiplication algorithm \cite{alman2025more} has $\omega = 2.371339$; however, this is a \emph{galactic} algorithm that would only be useful for $n$ too large for any present-day computer to handle.
The most practical of the sub-cubic algorithms is the original algorithm by \textcite{strassen1969gaussian}, and has $\omega = \log_{2}{7} \approx 2.8074$.
The Strassen algorithm is useful for large matrices (with $n$ exceeding $1000$ or so \cite{dalberto2005using,huang2016strassen}) over exact domains, where its reduced numerical stability \cite{miller1974computational,ballard2016improving} is not an issue.
Even so, the algorithm is often avoided due to its larger memory requirements (although asymptotically it has the same $\softO(n^{2})$ space complexity as naive matrix multiplication).

\subsection{Approximate random sampling}
\label{sec:random_sampling}

Rather than explicitly computing all of the necessary order intervals in a causal set, we can instead compute a random subset of these and use the distribution of their volumes to estimate the total abundances of different volumes.
Combining these estimates of abundances as in \cref{def:bd_action_4} leads to an estimate of the BDG action.

Let us take a random sample of $K \leq N$ causally related pairs of points.
Random variables $K_{k}$ for the number of $k$-nearest neighbor pairs in the sample can be used to estimate the $k$-abundances $N_{k}$ for the whole causal set: The random variables $K_{k} / K$ are unbiased estimators for $N_{k} / N$.
Thus, by \cref{def:bd_action_4}, an estimator $\hat{S}$ for the BDG action in four dimensions is given by
\begin{equation}
    \frac{1}{\hbar} \hat{S}
    = \frac{4}{\sqrt{6}} \mathopen{}\left( \frac{l}{l_{p}} \right)^{2} \mathopen{}\left( n - \frac{N}{K} K_{0} + 9 \frac{N}{K} K_{1} - 16 \frac{N}{K} K_{2} + 8 \frac{N}{K} K_{3} \right)\mathclose{} .
\end{equation}
The error in such an estimate can be quantified by its standard deviation $\sigma(\hat{S})$, a simple bound for which is given by
\begin{equation}
\label{eq:random_err_bound}
    \frac{1}{\hbar} \sigma(\hat{S})
    \leq 34 \frac{2}{\sqrt{6}} \mathopen{}\left( \frac{l}{l_{p}} \right)^{2} \frac{N}{\sqrt{K}} \sqrt{\frac{N-K}{N-1}} .
\end{equation}
See \cref{eq:random_se_full,eq:random_se_subadditive_bound} of \cref{sec:sampling_error} for tighter error bounds.
Under the bound of \cref{eq:random_err_bound}, it can be seen that if $K$ is at least a constant fraction of $N$ then $\sigma(\hat{S}) = \bigO(\sqrt{N}) = \bigO(n)$, which is the scaling required for the \emph{relative} error to remain constant.
Specifically, letting the sample size
\begin{equation}
    K = \frac{N}{1 + \varepsilon^{2} (1 - N^{-1})}
\end{equation}
for some choice of $\varepsilon \geq 0$, we find that
\begin{align}
    \frac{1}{\hbar} \sigma(\hat{S})
    & \leq 34 \frac{2}{\sqrt{6}} \mathopen{}\left( \frac{l}{l_{p}} \right)^{2} \varepsilon \sqrt{N} \\
    & < \frac{34 \varepsilon n}{\sqrt{3}} \mathopen{}\left( \frac{l}{l_{p}} \right)^{2} .
\end{align}
Since we need to compute $\bigO(\frac{n^{2}}{1 + \varepsilon^{2}})$ discrete volumes, each of which takes $\bigO(n)$ time, the overall time complexity is thus $\bigO(\frac{n^{3}}{1 + \varepsilon^{2}})$.
The naive exact method of \cref{sec:exact_methods} is simply the special case with relative error chosen to be $\varepsilon = 0$.
Similar analyses can be performed for the BDG action in other dimensions.

\section{Quantum algorithm}
\label{sec:quantum_approach}

We now propose a quantum algorithm for the BDG action.
The algorithm is split into multiple components performed in sequence, which we present in \cref{sec:components} before combining them in \cref{sec:combined_algorithm}.

\subsection{Algorithm components}
\label{sec:components}

There are three main components: The preparation of a uniform superposition of data imported from the classical description of the causal set (using techniques from \cref{sec:uniform_superposition,sec:mcx}), the construction of a specific data-independent quantum oracle circuit to be used in the quantum counting algorithm, and an adapted version of the quantum counting algorithm discussed in \cref{sec:quantum_counting} with suitable error scaling.

\subsubsection{Data preparation}
\label{sec:data_prep}

Given a causal set with $n$ elements, the goal of this section is to prepare a uniform superposition over the $n^{2}$-dimensional subspace of $2n$-qubit basis states whose binary strings are all $n^{2}$ possible concatenated pairs of rows and columns
\begin{equation}
    D = \{ \bm{r}_{i} \bm{c}_{j} \}_{i,j=1}^{n} ,
\end{equation}
from the adjacency matrix of the causal set, where binary strings $\bm{r}_{i}$ and $\bm{c}_{j}$ are the $i$th row and $j$th column of the (reflexive) causal set adjacency matrix $\bm{\bar{A}}$, respectively.
It would suffice in what follows to consider fewer pairs (a lower dimensional subspace) by removing those whose row and column correspond to the same causal set element (or if the causal set is given in topological order, then by only considering those pairs whose column corresponds to a point to the future of that of the row); however, this would not affect the asymptotic resource complexities of the overall algorithm, so we choose not to do so.

To create the desired state, we first use a $\lceil 2 \log_{2}{n} \rceil$-qubit register $\mathcal{C}$ to produce an initial uniform superposition over the correct number $n^{2}$ of basis states
\begin{equation}
\label{eq:initial_superposition}
    \ket{\phi} = \frac{1}{n} \sum_{j=0}^{n^2 - 1} \ket{j} .
\end{equation}
This can be achieved using a circuit of depth $\bigO(\log{n})$ (see \cref{sec:uniform_superposition}).
We then proceed to map this superposition onto a larger $2n$-qubit register $\mathcal{T}$ using multi-controlled multi-target \textsc{NOT} gates (see \cref{sec:mcx}), imbuing it with the causal set data in the process.
As we will see, the final result of the section (\cref{cor:uniform_superposition}) is a uniform superposition state of the form
\begin{equation}
    \frac{1}{n} \sum_{i,j=1}^{n} \ket{h(i,j)} \otimes \ket{\bm{r}_{i}} \ket{\bm{c}_{j}} ,
\end{equation}
where $h(i,j)$ are $n^{2}$ unique integers enumerating the pairs of rows and columns (as defined by \cref{eq:row_col_enumeration} in the following explanation).
The state is produced using a circuit of depth $\bigO(n^{2} \log\log{n})$.

Let us define a helpful alternative notation for multi-controlled multi-target \textsc{NOT} gates that specifies which qubits of our registers $\mathcal{C}$ and $\mathcal{T}$ are to be used as control and target qubits, respectively.
Denote by $n_{\mathcal{C}}$ the number of qubits comprising the register $\mathcal{C}$ and by $n_{\mathcal{T}}$ the number of qubits comprising the register $\mathcal{T}$ (in our scenario $n_{\mathcal{C}} = \lceil 2 \log_{2}{n} \rceil$ and $n_{\mathcal{T}} = 2n$).
For any $c \in \{0, \dots, 2^{n_{\mathcal{C}}} - 1\}$ and $t \in \{0, \dots, 2^{n_{\mathcal{T}}} - 1\}$, we let $X_{t}^{c}$ denote the gate acting on $\mathcal{C} \otimes \mathcal{T}$ defined as
\begin{equation}
\label{eq:mcmtx_binary}
    X_{t}^{c} =  \dyad{c} \otimes \textsc{NOT}_{t} + (I_{\mathcal{C}} - \dyad{c}) \otimes I_{\mathcal{T}} ,
\end{equation}
where with $t_{j}$ denoting the $j$th bit of the binary representation of $t$ and $\sigma_{\text{x}}$ denoting the Pauli-$X$ operator, we have also defined
\begin{equation}
    \textsc{NOT}_{t} = \bigotimes_{j=1}^{n_{\mathcal{T}}} \sigma_{\text{x}}^{t_{j}} .
\end{equation}

\begin{lemma}
\label{lem:mcmtx_depth}
    For any $c \in \{0, \dots, 2^{n_{\mathcal{C}}} - 1\}$ and $t \in \{0, \dots, 2^{n_{\mathcal{T}}} - 1\}$, the gate $X_{t}^{c}$ can be implemented as a quantum circuit with a depth of at most $\bigO(\log(n_{\mathcal{C}} n_{\mathcal{T}}))$ using $n_{\mathcal{C}} - 1$ ancilla qubits.
\end{lemma}
\begin{proof}
    If $w(x)$ denotes the Hamming weight (i.e. number of nonzero symbols) of a binary string $x$, then (aside from when $c=0$ or $t=0$) $X_{t}^{c}$ coincides with the $\textsc{C}^{w(c)}\textsc{NOT}^{w(t)}$ gate that has as its control qubits those in the register $\mathcal{C}$ whose positions correspond to the positions of $1$'s in the binary representation of $c$, and has as its target qubits those in the register $\mathcal{T}$ whose positions correspond to the positions of $1$'s in the binary representation of $t$.
    An example of this is given in \cref{fig:mcmtx_map}.
    In cases where $t=0$, we always have the identity $X_{0}^{c} = I_{\mathcal{C}} \otimes I_{\mathcal{T}}$ according to \cref{eq:mcmtx_binary}.
    In cases where both $c=0$ and $t \neq 0$, we have according to \cref{eq:mcmtx_binary} that $X_{t}^{0}$ activates $\textsc{NOT}_{t}$ on the target qubits when all control qubits have states $\ket{0}$.
    This can be implemented by applying Pauli-$X$ gates to every control qubit in $\mathcal{C}$ before and after a $\textsc{C}^{n_{\mathcal{C}}}\textsc{NOT}^{w(t)}$ gate with its target qubits determined as usual by the binary representation of $t$.
    Since the maximum possible number of control qubits is $n_{\mathcal{C}}$ and the maximum target qubits is $n_{\mathcal{T}}$, the asymptotic circuit depth and the number of ancilla required are the same as those of multi-controlled multi-target \textsc{NOT} gates (see \cref{sec:mcx}).
\end{proof}
\begin{remark}
    For a causal set of $n$ elements, this means that the circuit depth is at most $\bigO(\log{n})$ using $\lceil 2 \log_{2}{n} \rceil - 1$ ancilla qubits.
\end{remark}

\begin{figure}[htb]
    \centering
    \begin{equation*}
        \begin{quantikz}
             \lstick{$\mathcal{C}$} & \qwbundle{2} & \gate[2]{X_{13}^{3}} & \\
             \lstick{$\mathcal{T}$} & \qwbundle{4} & &
        \end{quantikz}
        =
        \begin{quantikz}
             \lstick{$\mathcal{C}$} & \qwbundle{2} & \gate[2]{X_{1101_{2}}^{11_{2}}} & \\
             \lstick{$\mathcal{T}$} & \qwbundle{4} & &
        \end{quantikz}
        =
        \begin{quantikz}
            \lstick[2]{$\mathcal{C}$} & \ctrl{5} & \\
            & \control{} & \\
            \lstick[4]{$\mathcal{T}$} & \targ{} & \\
            & \targ{} & \\
            & & \\
            & \targ{} &
        \end{quantikz}
    \end{equation*}
    \caption{
        Example multi-controlled multi-\textsc{NOT} gate $X_{13}^{3}$ with controls on a two-qubit register $\mathcal{C}$ and targets on a four-qubit register $\mathcal{T}$ (as would be the register setup for a causal set with just two elements).
        Since the binary representation of $3$ is $11_{2}$, both qubits of $\mathcal{C}$ are controls.
        Since the binary representation of $13$ is $1101_{2}$, qubits $1$, $2$, and $4$ are targets.
    }
    \label{fig:mcmtx_map}
\end{figure}
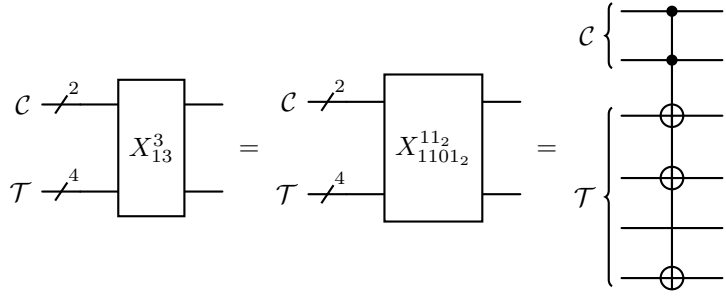

The gate $X_{t}^{c}$ acts on basis states $\ket{x}_{\mathcal{C}} \otimes \ket{y}_{\mathcal{T}}$ as
\begin{equation}
    X_{t}^{c} \ket{x}_{\mathcal{C}} \otimes \ket{y}_{\mathcal{T}} =
    \begin{cases}
        \ket{c}_{\mathcal{C}} \otimes \ket{y \oplus t}_{\mathcal{T}} & \text{if $x = c$,} \\
        \ket{x}_{\mathcal{C}} \otimes \ket{y}_{\mathcal{T}} & \text{otherwise.}
    \end{cases}
\end{equation}
Suppose that $D \subset \{0, 1\}^{m}$ is a set of $m$-bit binary strings (the \emph{data}) that we desire to be placed in a uniform superposition.
We can enumerate this data (starting from $0$) by any surjection $f \colon \{0, \dots, N - 1\} \to D$, where $N \geq \lvert D \rvert$.
Prepare the register $\mathcal{C}$ with $n_{\mathcal{C}} = \lceil \log_{2}{N} \rceil$ qubits in the state
\begin{equation}
\label{eq:general_initial_superposition}
    \ket{\phi}_{\mathcal{C}} = \frac{1}{\sqrt{N}} \sum_{j=0}^{N-1} \ket{j}_{\mathcal{C}}
\end{equation}
and prepare the register $\mathcal{T}$ with $n_{\mathcal{T}} = m$ qubits in state $\ket{0}_{\mathcal{T}}$.
We can then apply the following sequence of gates resulting in
\begin{equation}
\label{eq:general_superposition}
    X_{f(0)}^{0} \dots X_{f(N-1)}^{N-1} \ket{\phi}_{\mathcal{C}} \otimes \ket{0}_{\mathcal{T}}
    = \frac{1}{\sqrt{N}} \sum_{j=0}^{N-1} \ket{j}_{\mathcal{C}} \otimes \ket{f(j)}_{\mathcal{T}} .
\end{equation}

\begin{theorem}
\label{thm:uniform_superposition}
    The uniform superposition of \cref{eq:general_superposition} can be prepared by applying a circuit of depth $\bigO(N \log(m \log{N}))$ using $\lceil \log_{2}{N} \rceil - 1$ reusable ancilla qubits, where all qubits are initially prepared as $\ket{0}$.
    The total width is $m + 2 \lceil \log_{2}{N} \rceil - 1$ qubits.
\end{theorem}
\begin{proof}
    Starting from the zeroed $\lceil \log_{2}{N} \rceil$-qubit state $\ket{0}_{\mathcal{C}}$, the initial uniform superposition $\ket{\phi}_{\mathcal{C}}$ of \cref{eq:general_initial_superposition} can be prepared using a circuit of depth $\bigO(\log{N})$ (see \cref{sec:uniform_superposition}).
    After this has been completed, each application of an $X_{f(j)}^{j}$ gate contributes an additional depth at most $\bigO(\log(m \log{N}))$ and makes use of $\lceil \log_{2}{N} \rceil - 1$ ancilla qubits (\cref{lem:mcmtx_depth}).
    Such gates are applied $N$ times as shown in \cref{eq:general_superposition}, and thus the overall depth becomes $\bigO(N \log(m \log{N}))$.
    All ancilla qubits are reset to $\ket{0}$ at the end of each application of a $X_{f(j)}^{j}$ gate and can then be reused, therefore no more than $\lceil \log_{2}{N} \rceil - 1$ ancilla qubits are needed.
\end{proof}

Returning to the consideration of our causal set, the data set that we import is every possible pair of rows and columns of its (reflexive) adjacency matrix $\bm{\bar{A}} = \bm A + \bm I$ (see \cref{sec:volume_matrix}), represented as $2n$-bit binary strings (concatenated row and column strings).
We thus set $m = 2n$.
Denoting binary strings representing values contained in the $i$th row of $\bm{\bar{A}}$ by $\bm{r}_{i}$ and in the $j$th column by $\bm{c}_{j}$, the data set is
\begin{equation}
    D = \{ \bm{r}_{i} \bm{c}_{j} \}_{i,j=1}^{n} .
\end{equation}
Since no two rows or columns of $\bm{\bar{A}}$ are the same (by the \emph{asymmetry} property of \cref{def:causet}), the set has cardinality $\lvert D \rvert = n^{2}$.
Let us thus set $N = n^{2}$.
For convenience, let us assign to each coordinate $(i,j)$ of an $n \times n$ matrix a unique integer by defining a bijective function $k \colon \{1, \dots, n\}^{2} \to \{0, \dots, n^{2} - 1\}$ as
\begin{equation}
\label{eq:row_col_enumeration}
    h(i,j) = (i-1)n + j-1
\end{equation}
which has inverse $h \mapsto (i(h), j(h))$, where
\begin{equation}
    i(h) = 1 + \mathopen{}\left\lfloor \frac{h}{n} \right\rfloor\mathclose{} ,\quad
    j(h) = 1 + h - n \mathopen{}\left\lfloor \frac{h}{n} \right\rfloor\mathclose{} .
\end{equation}
\begin{remark}
    The assignment by $h$ is simply incrementing an integer starting at $0$ by one for each element of an $n \times n$ matrix representation when read starting from the top left and going first from left-to-right and then top-to-bottom.
    For $n=3$, for example, $h(i,j) = H_{ij}$, where
    \begin{equation}
        \bm{H} =
        \begin{pmatrix}
            0 & 1 & 2 \\
            3 & 4 & 5 \\
            6 & 7 & 8
        \end{pmatrix} .
    \end{equation}
\end{remark}
We can then define a bijection $f \colon \{0, \dots, n^{2} - 1\} \to D$ by
\begin{equation}
    f(h) = \bm{r}_{i(h)} \bm{c}_{j(h)} .
\end{equation}
Preparing $\ket{\phi}_{\mathcal{C}}$ and applying the gates as in \cref{eq:general_superposition} with our choices of $m$, $D$, $N$, and $f$ yields the following desired corollary to \cref{thm:uniform_superposition}.

\begin{corollary}
\label{cor:uniform_superposition}
    The uniform superposition
    \begin{equation}
        \frac{1}{n} \sum_{i,j=1}^{n} \ket{h(i,j)}_{\mathcal{C}} \otimes \ket{\bm{r}_{i}} \ket{\bm{c}_{j}}
    \end{equation}
    can be prepared by applying a circuit of depth $\bigO(n^{2} \log{n})$ using $\lceil 2 \log_{2}{n} \rceil - 1$ reusable ancilla qubits, where all qubits are initially prepared as $\ket{0}$.
    The total width is $2(n + \lceil 2 \log_{2}{n} \rceil) - 1$ qubits.
\end{corollary}

\subsubsection{Oracle circuit}
\label{sec:bd_oracle}

Recall from \cref{sec:volume_matrix} that $k$-abundances can be obtained from the (reflexive) adjacency matrix $\bm{\bar{A}}$ of a causal set with $n$ elements.
Given $n$-bit binary strings $\bm r$ and $\bm c$ encoded as a $2n$-qubit quantum state $\ket{\bm r} \ket{\bm c}$, we wish to count the number of positions where the two strings share an entry $1$ (i.e. compute the Hamming weight of their bitwise \textsc{AND}, or equivalently their dot product).
If $\bm r$ represents the values contained in an arbitrary row of $\bm{\bar{A}}$ and $\bm c$ of an arbitrary column, then $\bm r \cdot \bm c$ is the inclusive discrete volume between elements in the causal set corresponding to this row and column.
We then wish to check if this computed volume is equal to $k+2$ for some choice of $0 \leq k \leq n_{d} - 1$ (where, for example, $n_{d} = 4$ for the BDG action in four dimensions).

\begin{figure}[htb]
    \centering
    \begin{equation*}
        \begin{quantikz}
            & \qwbundle{q} & \gate{\textsc{INCR}} &
        \end{quantikz}
        =
        \begin{quantikz}[wire types={q,n,q,q,q}]
            & \targ{} &&&&& \\
            \wave &&&&&& \\
            & \control{} & \ \ldots \ & \targ{} &&& \\
            & \control{} & \ \ldots \ & \control{} & \targ{} && \\
            & \ctrl{-4} & \ \ldots \ & \ctrl{-2} & \ctrl{-1} & \targ{} &
        \end{quantikz}
    \end{equation*}
    \begin{equation*}
        \begin{quantikz}
            & \qwbundle{q} & \gate{\textsc{DECR}} &
        \end{quantikz}
        =
        \begin{quantikz}[wire types={q,n,q,q,q}]
            &&&&& \targ{} & \\
            \wave &&&&&& \\
            &&& \targ{} & \ \ldots \ & \control{} & \\
            && \targ{} & \control{} & \ \ldots \ & \control{} & \\
            & \targ{} & \ctrl{-1} & \ctrl{-2} & \ \ldots \ & \ctrl{-4} &
        \end{quantikz}
    \end{equation*}
    \caption{
        Definitions of quantum gates that increment and decrement qubit representations of integers by one.
        Note that $\textsc{INCR}^{-1} = \textsc{DECR}$.
        Handling binary representations of up to an integer $n$ requires $q = 1 + \lfloor \log_{2}{n} \rfloor$ qubits and multi-controlled \textsc{NOT} gates.
        Since $m$-controlled \textsc{NOT} gates can be decomposed into $\bigO(\log(m)^{3})$ single-qubit and \textsc{CNOT} gates (see \cref{sec:mcx}) using a single reusable ancilla qubit, both \textsc{INCR} and \textsc{DECR} gates can be decomposed with depth $\bigO(q \log(q)^{3}) = O(\log(n) (\log\log{n})^{3})$.
    }
    \label{fig:incr_decr}
\end{figure}

\Cref{fig:incr_decr} defines quantum circuits whose actions are incrementing and decrementing binary representations of integers (strings of computational basis states) by one.
We now use these for counting when computing dot products.
A complete quantum oracle circuit to check whether the volume is $k+2$ as chosen is given in \cref{fig:bd_oracle}.
This oracle $V_{k}$ acts on a computational basis state $\ket{\bm r} \ket{\bm c}$ to introduce a negative phase
\begin{equation}
    V_{k} \ket{\bm r} \ket{\bm c} =
    \begin{cases}
        - \ket{\bm r} \ket{\bm c} & \text{if $\bm r \cdot \bm c = k+2$,} \\
        \ket{\bm r} \ket{\bm c} & \text{otherwise.}
    \end{cases}
\end{equation}

\begin{figure}[htb]
    \centering
    \begin{adjustbox}{width=\textwidth}
    \begin{quantikz}[wire types={q,q,n,q,q,q,n,q,q,q,q,n,q,q,q,n}]
        \lstick[4]{$\ket{\bm r}$ \\ input row \\ $n$ qubits} & \ctrl{8} && \ctrl{8} &&&&&&&&&&&&&&&&&&&&&&& \ctrl{8} && \ctrl{8} & \\
        &&&& \ctrl{7} && \ctrl{7} &&&&&&&&&&&&&&&&& \ctrl{7} && \ctrl{7} &&&& \\
        \wave &&&&&&&&&&&&&&&&&&&&&&&&&&&&& \\
        &&&&&&&& \ctrl{5} && \ctrl{5} &&&&&&&&& \ctrl{5} && \ctrl{5} &&&&&&&& \\
        \lstick[4]{$\ket{\bm c}$ \\ input col. \\ $n$ qubits} & \control{} && \control{} &&&&&&&&&&&&&&&&&&&&&&& \control{} && \control{} & \\
        &&&& \control{} && \control{} &&&&&&&&&&&&&&&&& \control{} && \control{} &&&& \\
        \wave &&&&&&&&&&&&&&&&&&&&&&&&&&&&& \\
        &&&&&&&& \control{} && \control{} &&&&&&&&& \control{} && \control{} &&&&&&&& \\
        \lstick{\textsc{AND} qubit $\ket{0}$} & \targ{} & \ctrl{4} & \targ{} & \targ{} & \ctrl{4} & \targ{} & \ \dots \ & \targ{} & \ctrl{4} & \targ{} &&&&&&&&& \targ{} & \ctrl{4} & \targ{} & \ \dots \ & \targ{} & \ctrl{4} & \targ{} & \targ{} & \ctrl{4} & \targ{} & \\
        \lstick[4]{$\ket{0 \dots 0}$ \\ inclusive volume \\ $1 + \lfloor \log_{2}{n} \rfloor$ qubits} && \gate[4]{\textsc{INCR}} &&& \gate[4]{\textsc{INCR}} && \ \dots \ && \gate[4]{\textsc{INCR}} &&& \ctrl{4} &&&&&& \ctrl{4} && \gate[4]{\textsc{DECR}} && \ \dots \ && \gate[4]{\textsc{DECR}} &&& \gate[4]{\textsc{DECR}} && \\
        &&&&&&& \ \dots \ &&&&&& \ctrl{4} &&&& \ctrl{4} &&&&& \ \dots \ &&&&&&& \\
        \wave &&&&&&&&&&&&&&&&&&&&&&&&&&&&& \\
        &&&&&&& \ \dots \ &&&&&&& \ctrl{4} && \ctrl{4} &&&&&& \ \dots \ &&&&&&& \\
        \lstick[4]{$\ket{k+2}$ \\ $k$-abundance param. \\ $1 + \lfloor \log_{2}{n} \rfloor$ qubits} &&&&&&&&&&&& \targ{} &&& \octrl{4} &&& \targ{} &&&&&&&&&&& \\
        &&&&&&&&&&&&& \targ{} && \ocontrol{} && \targ{} &&&&&&&&&&&& \\
        \wave &&&&&&&&&&&&&&&&&&&&&&&&&&&&& \\
        &&&&&&&&&&&&&& \targ{} & \ocontrol{} & \targ{} &&&&&&&&&&&&& \\
        \lstick{phase qubit $\ket{-}$} &&&&&&&&&&&&&&& \targ{} \slice{uncompute} &&&&&&&&&&&&&&
    \end{quantikz}
    \end{adjustbox}
    \caption{
        Quantum oracle circuit $V_{k}$ for the $N_{k}$ term of the BDG action for a causal set with $n$ elements.
        This circuit has width $O(n)$ qubits (which are $2n$ row/column data qubits and $2(\lfloor \log_{2}{n} \rfloor + 2)$ ancilla qubits) and depth $\tilde{O}(n)$ in single-qubit and \textsc{CNOT} gates.
        The \textsc{AND} qubit and phase qubits can be used as the borrowed ancilla required to implement multi-controlled \textsc{CNOT} gates used in the circuit (e.g., those underlying the \textsc{INCR} and \textsc{DECR} gates) and so do not affect the width.
        The cost of uncomputation can be reduced by a measure-and-fixup approach, however this does not affect the asymptotic complexity \cite{jones2013low,gidney2018halving}.
    }
    \label{fig:bd_oracle}
\end{figure}
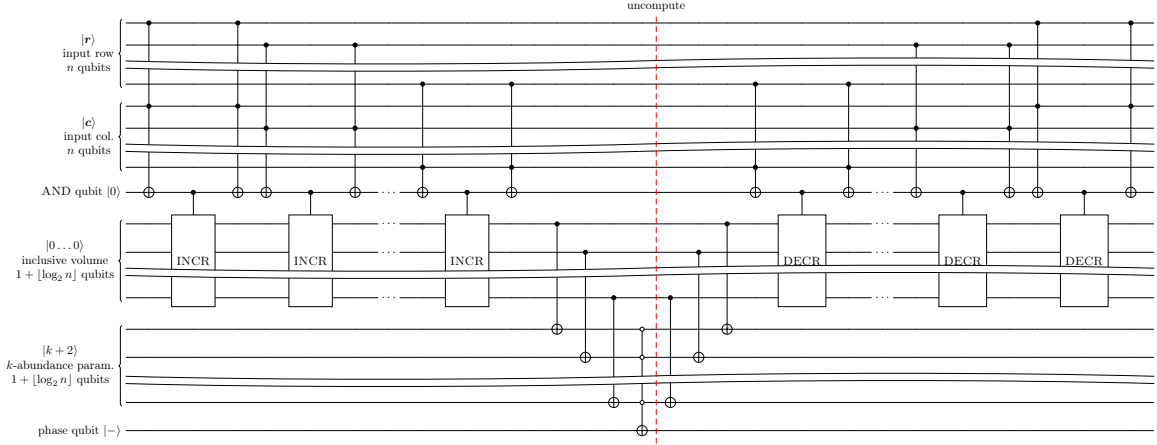

\begin{theorem}
\label{thm:bd_oracle}
    The oracle circuit $V_{k}$ of \cref{fig:bd_oracle}, which checks whether the dot product of two $n$-bit binary strings is equal to $k+2$, has depth $\bigO(n \log(n) (\log\log{n})^{3})$ and makes use of $2(\lfloor \log_{2}{n} \rfloor + 2)$ ancilla qubits.
    The total width is $2(n + \lfloor \log_{2}{n} \rfloor + 2)$.
\end{theorem}
\begin{proof}
    The circuit proceeds by checking the first qubit of each of $\ket{r}$ and $\ket{c}$ and, if they both take value $1$, storing this \textsc{AND} to a zeroed single-qubit ancilla and incrementing another ancilla register storing a $1 + \lfloor \log_{2}{n} \rfloor$-bit integer (initially prepared as $\ket{0 \dots 0}$) by one before returning the \textsc{AND} ancilla to zero.
    This process is then repeated for the second qubits of $\ket{r}$ and $\ket{c}$, and so on until the $n$th qubits.
    Since the increment gates \textsc{INCR} each have depth $\bigO(\log(n) (\log\log{n})^{3})$, the depth of all $n$ of these steps performed in sequence is $\bigO(n \log(n) (\log\log{n})^{3})$.
    Our register of ancilla qubits now contains a binary representation of the inclusive discrete volume.
    The circuit now performs \textsc{CNOT} gates controlled on each qubit of the ancilla and targeting each qubit of a further ancilla register of equal size prepared in the state $\ket{k+2}$.
    This has depth $\bigO(1)$, as all gates are \textsc{CNOT} and can be performed simultaneously.
    If the qubits of both ancillae represent the same integer ($k+2$), then the $\ket{k+2}$-prepared ancilla now lies in the state $\ket{0 \dots 0}$.
    A single $\textsc{C}^{1 + \lfloor \log_{2}{n} \rfloor}\textsc{NOT}$ gate (with controls reversed by Pauli-$X$ gates) is then applied to introduce the negative phase to a $\ket{-}$ qubit ancilla if this is the case, which has depth $\bigO((\log\log{n})^{3})$.
    Finally, we perform uncomputation to return all qubits back to their original state (but with the possible negative phase).
    The total depth of the circuit is therefore $\bigO(n \log(n) (\log\log{n})^{3})$, with a width of $2(n + \lfloor \log_{2}{n} \rfloor + 2)$ qubits.
\end{proof}

\subsubsection{Counting algorithm}
\label{sec:abundance_counting_alg}

In order to ensure that the error of our BDG action estimate scales linearly in $n$ rather than quadratically, we would like a version of quantum counting (\cref{thm:approx_counting}) such that the error is bounded by $\bigO(\varepsilon \sqrt{K})$ for some $\varepsilon > 0$.
We can apply two runs of the algorithm in \cref{thm:approx_counting} to achieve this.
Setting the parameter $\delta = 1 - \sqrt{1 - \zeta}$, perform the following steps.
\begin{enumerate}
    \item Run quantum counting with the error parameter $\epsilon = \varepsilon$ to estimate $K$ by $\hat{K}$.

    \item Use this estimate $\hat{K}$ to run the algorithm again with appropriately reduced error parameter
    \begin{equation}
    \epsilon = \varepsilon \sqrt{1 - \varepsilon} \hat{K}^{- \frac{1}{2}} .
    \end{equation}
\end{enumerate}

\begin{theorem}[Approximate counting with square-root error]
\label{thm:approx_counting_sqrt_err}
    Let $f \colon \{0, \dots, N-1\} \to \{0, 1\}$ be a Boolean function marking $K = \lvert f^{-1}(1) \rvert > 0$ items.
    Given access to a membership oracle for $f$ and $\varepsilon, \zeta > 0$, there exists a quantum algorithm that outputs an estimate $\hat{K}$ for $K$.
    The output $\hat{K}$ satisfies
    \begin{equation}
        \lvert \hat{K} - K \rvert < \varepsilon \sqrt{K}
    \end{equation}
    with probability at least $1 - \zeta$.
    The algorithm makes fewer than
    \begin{equation}
        \bigO \mathopen{}\left( \varepsilon^{-1} \sqrt{N} \log{\frac{1}{\zeta}} \right)\mathclose{}
    \end{equation}
    oracle queries.
    The algorithm requires $\bigO(\log{N})$ qubits of space.
\end{theorem}
\begin{proof}
    After the first run of the algorithm of \cref{thm:approx_counting} with error parameter $\epsilon_{1} = \varepsilon$, we obtain an estimate $\hat{K}_{1}$ for $K$ that satisfies
    \begin{equation}
    \label{eq:K1_estimate}
        (1 - \varepsilon) K < \hat{K}_{1} < (1 + \varepsilon) K
    \end{equation}
    with probability at least $1 - \delta$.
    We now run the algorithm again, but this time we set the error parameter to
    \begin{equation}
    \label{eq:eps2}
        \epsilon_{2} = \varepsilon \sqrt{1 - \varepsilon} \hat{K}_{1}^{- \frac{1}{2}} .
    \end{equation}
    Inserting \cref{eq:K1_estimate} into \cref{eq:eps2} gives
    \begin{equation}
    \label{eq:eps2_bound}
        \sqrt{\frac{1 - \varepsilon}{1 + \varepsilon}} \frac{\varepsilon}{\sqrt{K}} < \epsilon_{2} < \frac{\varepsilon}{\sqrt{K}} .
    \end{equation}
    Thus, the estimate $\hat{K}$ obtained from the second run of the algorithm satisfies $\lvert \hat{K} - K \rvert < \epsilon_{2} K < \varepsilon \sqrt{K}$ with a probability of at least $1 - \delta$, conditioned on the fact that $\hat{K}_{1}$ satisfied the error inequality of the first run.
    The overall probability that the error bound for $\hat{K}$ is satisfied is therefore at least $(1 - \delta)^{2} = 1 - \zeta$.

    The first run of the algorithm makes fewer than $\bigO \mathopen{}\bigl( \varepsilon^{-1} \sqrt{N / K} \log(\delta^{-1}) \bigr)\mathclose{}$ oracle queries.
    Since we have from the lower bound of \cref{eq:eps2_bound} that
    \begin{equation}
        \epsilon_{2}^{-1}
        < \sqrt{\frac{1 + \varepsilon}{1 - \varepsilon}} \frac{\sqrt{K}}{\varepsilon}
        \leq \frac{\sqrt{K}}{\varepsilon(1 - \varepsilon)} ,
    \end{equation}
    the second run makes fewer than $\bigO \mathopen{}\bigl( \varepsilon^{-1} \sqrt{N} \log(\delta^{-1}) \bigr)\mathclose{}$ queries.
    The total number of queries of both runs is therefore also $\bigO \mathopen{}\bigl( \varepsilon^{-1} \sqrt{N} \log(\delta^{-1}) \bigr)\mathclose{}$.
    By our choice of $\delta = 1 - \sqrt{1 - \zeta}$, we have $\log_{2}(\delta^{-1}) \leq 1 + \log_{2}(\zeta^{-1})$ and the result follows.
\end{proof}

\subsection{Combined algorithm}
\label{sec:combined_algorithm}

We now have all the pieces in place to describe the algorithm to estimate the $k$-abundance $N_{k}$ for a causal set with $n$ elements.
\begin{enumerate}
    \item Prepare the uniform superposition state
    \begin{equation}
        \frac{1}{n} \sum_{i,j=1}^{n} \ket{h(i,j)}_{\mathcal{C}} \otimes \ket{\bm{r}_{i}} \ket{\bm{c}_{j}}
    \end{equation}
    as described in \cref{sec:data_prep}, where the $\bm{r}_{i}$ and $\bm{c}_{j}$ encode the $i$th rows and $j$th columns of the reflexive adjacency matrix of the causal set as $n$-bit binary strings.
    \begin{description}
        \item[Depth] $\bigO(n^{2} \log{n})$.
        \item[Width] $2n + 2 \lceil 2 \log_{2}{n} \rceil - 1$ qubits.
    \end{description}

    \item Perform quantum counting (see \cref{thm:approx_counting} of \cref{sec:quantum_counting}) with some choice $\varepsilon > 0$ (pertaining to the tolerable error), probability of failure $\delta$, and oracle as constructed in \cref{fig:bd_oracle} of \cref{sec:bd_oracle} over the subspace spanned by the $n^2$ orthonormal states $\ket{h(i,j)}_{\mathcal{C}} \otimes \ket{\bm{r}_{i}} \ket{\bm{c}_{j}}$.
    \begin{description}
        \item[Error] $\lvert \hat{N_{k}} - N_{k} \rvert < \varepsilon \sqrt{N_{k}} < \frac{\varepsilon n}{\sqrt{2}}$.
        \item[Depth] $\bigO \mathopen{}\left( n^{2} \log(n) (\log\log{n})^{3} \varepsilon^{-1} \log(\delta^{-1}) \right)\mathclose{}$.
        \item[Width] $2n + \lceil 2 \log_{2}{n} \rceil + 2 \lfloor \log_{2}{n} \rfloor + 4 = \bigO(n)$ qubits.
    \end{description}
    The upper error bound is due to $k$-abundances being unable to exceed the total number of pairs of causal set elements $N_{k} \leq \frac{1}{2}n(n-1) < \frac{1}{2} n^{2}$.
    The extra $\lceil 2 \log_{2}{n} \rceil$ term in the width compared to \cref{thm:bd_oracle} appears due to the qubits in the register $\mathcal{C}$ that enumerate the data states prepared in the last stem, and must be held for the duration of the quantum counting without being measured.
\end{enumerate}
The total depth of both state preparation and counting is $\softO(n^{2} \varepsilon^{-1})$.
The total width of the protocol is $2n + \lceil 2 \log_{2}{n} \rceil + 2 \lfloor \log_{2}{n} \rfloor + 4 = \bigO(n)$ since the ancilla qubits of the first step can be reused for the oracle circuit in the second step.

As discussed at the beginning of \cref{sec:problem}, to compute an estimate of the BDG action we can simply repeat this as a subroutine once for each $N_{k}$ appearing in the action (\cref{def:bd_action_d}) and then combine the estimated $N_{k}$ accordingly.
Thus, $n_{d}$ runs must be performed to estimate the BDG action in $d$ dimensions.
The error in the resulting estimate of the action propagates from the errors in the $N_{k}$ (each of which is correct with probability at least $\delta$) by the triangle inequality, with the overall success probability becoming at least $(1 - \delta)^{n_{d}}$.
For example, we can estimate the BDG action in four dimensions (\cref{def:bd_action_4}) with error
\begin{equation}
\begin{alignedat}{2}
    \frac{1}{\hbar} \lvert \hat{S} - S \rvert
    & \leq \frac{4}{\sqrt{6}} \mathopen{}\left( \frac{l}{l_{p}} \right)^{2}
    \mathopen{}\Bigl[ && \lvert \hat{N}_{0} - N_{0} \rvert + 9 \lvert \hat{N}_{1} - N_{1} \rvert \\
    &&& + 16 \lvert \hat{N}_{2} - N_{2} \rvert + 8 \lvert \hat{N}_{3} - N_{3} \rvert \Bigr]\mathclose{} \\
    & < \frac{68 \varepsilon n}{\sqrt{3}} \mathopen{}\left( \frac{l}{l_{p}} \right)^{2}
\end{alignedat}
\end{equation}
occurring with probability at least $(1 - \delta)^{4}$.
Finally, we can state the following result of the algorithm we have described.

\begin{theorem}
\label{thm:bd_quantum_algorithm}
    Let $C$ be a causal set of $n$ elements.
    Given a spacetime dimension $d$ and $\delta, \varepsilon > 0$, the quantum algorithm of \cref{alg:bd_quantum_algorithm} outputs an estimate $\hat{S}^{(d)}(C)$ for the Benincasa--Dowker--Glaser action $S^{(d)}(C)$.
    The estimate satisfies
    \begin{equation}
        \left\lvert \hat{S}^{(d)}(C) - S^{(d)}(C) \right\rvert < \bigO(\varepsilon n)
    \end{equation}
    with probability at least $(1 - \delta)^{n_{d}}$.
    The algorithm has a running time
    \begin{equation}
        \softO \mathopen{}\left( n^{2} \varepsilon^{-1} \log(\delta^{-1}) \right)\mathclose{}
    \end{equation}
    and requires
    \begin{equation}
        2n + \lceil 2 \log_{2}{n} \rceil + 2 \lfloor \log_{2}{n} \rfloor + 4 = \bigO(n)
    \end{equation}
    qubits of space.
\end{theorem}

\section{Discussion}
\label{sec:discussion}

We have constructed an efficient algorithm for computing the Benincasa--Dowker--Glaser discrete causal set action in $d$ dimensions.
Starting from a classical description of a causal set with $n$ elements, we showed that all abundances of $k$-nearest neighbor pairs (and thus the BDG action itself) can be found in time $\softO(n^{2})$ by quantum counting.

Expressed here as $\softO(n^{2})$, we have suppressed the logarithmic factors that appear in the time complexity.
These factors are inherited from the time complexity of implementing multi-controlled \textsc{NOT} and fan-out gates.
If improvements are made to state-of-the-art implementations of these gates, either in their compilation (in terms of single- and two-qubit gates) or specialized hardware architecture, then these can be used directly to improve the time complexity of our protocol \cite{broadbent2009parallelizing,pham20132d,rasmussen2020single,guo2022implementing,fenner2023implementing,claudon2024polylogarithmic}.
However, we expect that BDG action algorithms have a time complexity of at least $\Omega(n^{2})$, since this is the number of possible pairs of elements in a causal set (and matches the best possible bound on $n \times n$ matrix multiplication algorithms with classical output).
In other words, we suspect that our runtime is tight up to logarithmic factors.
In terms of space complexity, our algorithm requires a total width of $\bigO(n)$ qubits.
This is because, unlike the naive classical $\bigO(1)$ algorithm (which can be performed with $\softO(1)$ space), our quantum algorithm must work with states encoding full $n$-bit rows and columns of the causal set adjacency matrix so that it can make use of the quantum ``parallelism'' associated with a superposition of data in order to perform the counting efficiently.
The width excluding these data qubits is also $\softO(1)$.
Our algorithm estimates the true value of the BDG action with an error proportional to $\varepsilon n$, where $\varepsilon$ is a chosen constant.
Although this is increasing in $n$, the growth is linear.
Since the action of the spacetime region over which the action is evaluated generally also scales linearly with $n$, the relative error (compared to its exact value) in our estimate is generally approximately constant and proportional to the choice of $\varepsilon$.
A comparison between merit figures for our quantum algorithm and other classical algorithms can be found in \cref{tab:comparison}.

\begin{table}[htb]
    \centering
    \begin{tabular}{lccc}
        \toprule
        Method & Absolute error & \multicolumn{2}{c}{Complexity} \\
        \cmidrule(l){3-4}
        && Time & Space \\
        \midrule
        Quantum counting & $\bigO(\varepsilon n)$ & $\softO(n^{2} \varepsilon^{-1})$ & $\bigO(n)$ \\
        Strassen matrix & Exact & $\bigO(n^{2.8074})$ & $\softO(n^{2})$ \\
        Random sampling & $\bigO(\varepsilon n)$ & $\bigO(n^{3} (1 + \varepsilon^{2})^{-1})$ & $\softO(1)$ \\
        Naive & Exact & $\bigO(n^{3})$ & $\softO(1)$ \\
        \bottomrule
    \end{tabular}
    \caption{
        Comparison between time and space resources required by the different methods discussed for computing the BDG action for a causal set with $n$ elements, either exactly or approximately, ordered by time complexity from fastest to slowest.
        While $\varepsilon$ represents the relative error in approximate methods, in order to achieve the same absolute error its value may differ by a constant factor between rows.
        For our quantum algorithm, the space complexity is the quantum width measured in number of qubits.
        We do not expect better than $\bigO(n^{2})$ time complexity to be achievable since there are this many pairs of elements in a causal set.
    }
    \label{tab:comparison}
\end{table}

\paragraph{Future work}

In certain practical applications (such as Monte Carlo experiments), it is not the BDG action expressed in \cref{def:bd_action_4} that is usually calculated, but rather a ``smeared'' or ``nonlocal'' BDG action that avoids large fluctuations from the mean by introducing a mesoscale (and associated nonlocality parameter) while still converging to the EH action in the continuum limit \cite{benincasa2010scalar,surya2012evidence,dowker2013causal,glaser2016hartle,henson2017onset,glaser2018finite}.
From an algorithmic perspective, the expressions are similar, but include contributions from larger order intervals (with the number of terms increasing with the causal set size).
However, by truncating this sum appropriately, our techniques for efficiently finding $k$-abundances may be directly applicable and result in a similar quantum speedup.

Another avenue for future work is that of improving the time complexity of our algorithm.
This could conceivably be achieved by more careful consideration of the data structure used when preparing our input state, utilizing any further state-of-the-art advances in compiling the generalized Toffoli gates that we use, or the application of specialized hardware architectures designed for these purposes \cite{broadbent2009parallelizing,pham20132d,rasmussen2020single,guo2022implementing,fenner2023implementing,claudon2024polylogarithmic,khattar2025rise}.

Although the quantum algorithm we have proposed requires large-scale fault-tolerant quantum hardware that has yet to be produced, we may consider whether a quantum counting variant (and therefore our algorithm) may be implemented in the near term using so-called noisy intermediate-scale quantum (NISQ) hardware \cite{brooks2019beyond}.
For example, it is known that the quadratic speedup of Grover's algorithm for unstructured search can be recovered on adiabatic quantum hardware, provided that their Hamiltonian evolution is optimally scheduled \cite{roland2002quantum,adamson2025adiabatic}.
There exist various schemes for adapting such adiabatic algorithms to NISQ hardware \cite{garcia2018addressing,harwood2022improving,kolotouros2024simulating}.
We may also estimate the resources required to perform the algorithm on fault-tolerant hardware, which include the necessary quantum error correcting codes \cite{beverland2022assessing}.
This resource estimation would allow us to better understand the timescale within which the fault-tolerant version of our algorithm could be of practical use.

\section*{Acknowledgments}

We would like to thank David A. Meyer and Adithya Sireesh for useful discussions.
S.A.A. and P.W. acknowledge funding from STFC with grant number ST/W006537/1, and EPSRC grants EP/T001062/1 and EP/Z53318X/1.

\appendix

\section{Standard error for random sampling}
\label{sec:sampling_error}

Consider randomly sampling (without replacement) $K$ pairs of points $x \prec y$ from a population of $N = \lvert E \rvert \leq n(n-1)/2$ pairs in a causal set with $n$ elements.
For any pair, we may compute the cardinality of their inclusive order interval (the inclusive \emph{volume}).
The points are called $k$-nearest neighbors if and only if this cardinality is $k+2$.

Let $N_{k}$ denote the total number of $k$-nearest neighbor pairs in the causal set and let $K_{k}$ denote a random variable for the number counted in the sample.
The number of pairs with each different volume in the sample $(K_{0}, \dots, K_{n-2})$ has the multivariate hypergeometric distribution.
We can then write that
\begin{subequations}
\begin{align}
    \Var(K_{i}) &= K \frac{N - K}{N - 1} \frac{N_{i}}{N} \mathopen{}\left( 1 - \frac{N_{i}}{N} \right)\mathclose{} , \\
    \Cov(K_{i}, K_{j}) &= - K \frac{N - K}{N - 1} \frac{N_{i}}{N} \frac{N_{j}}{N} \leq 0 .
\end{align}
\end{subequations}

An unbiased estimator for $N_{k} / N$ is the proportion of $k$-element order intervals $K_{k} / K$ in the sample.
The action
\begin{equation}
    \frac{1}{\hbar} S = \frac{4}{\sqrt{6}} \mathopen{}\left( \frac{l}{l_{p}} \right)^{2} (n - N_{0} + 9 N_{1} - 16 N_{2} + 8 N_{3})
\end{equation}
can thus be estimated by
\begin{equation}
    \frac{1}{\hbar} \hat{S}
    = \frac{4}{\sqrt{6}} \mathopen{}\left( \frac{l}{l_{p}} \right)^{2} \mathopen{}\left( n - \frac{N}{K} K_{0} + 9 \frac{N}{K} K_{1} - 16 \frac{N}{K} K_{2} + 8 \frac{N}{K} K_{3} \right)\mathclose{} .
\end{equation}
The variance of this estimate is
\begin{equation}
    \frac{1}{\hbar^{2}} \Var(\hat{S})
    = \frac{8}{3} \mathopen{}\left( \frac{l}{l_{p}} \right)^{4} \mathopen{}\left( \frac{N}{K} \right)^{2} \Var(- K_{0} + 9 K_{1} - 16 K_{2} + 8 K_{3}) ,
\end{equation}
where the variance of the linear combination can be expressed as
\begin{equation}
\begin{split}
    &\Var(- K_{0} + 9 K_{1} - 16 K_{2} + 8 K_{3}) \\
    &\qquad = \Var(K_{0}) + 81 \Var(K_{1}) + 256 \Var(K_{2}) + 64 \Var(K_{3}) \\
    &\qquad\qquad - 18 \Cov(K_{0}, K_{1}) + 32 \Cov(K_{0}, K_{2}) - 16 \Cov(K_{0}, K_{3}) \\
    &\qquad\qquad - 288 \Cov(K_{1}, K_{2}) + 144 \Cov(K_{1}, K_{3}) - 256 \Cov(K_{2}, K_{3}) \\
    &\qquad = K \frac{N - K}{N - 1} \mathopen{}\biggl[
    \frac{N_{0}}{N} \mathopen{}\left( 1 - \frac{N_{0}}{N} \right)\mathclose{}
    + 81 \frac{N_{1}}{N} \mathopen{}\left( 1 - \frac{N_{1}}{N} \right)\mathclose{}
    + 256 \frac{N_{2}}{N} \mathopen{}\left( 1 - \frac{N_{2}}{N} \right)\mathclose{}
    + 64 \frac{N_{3}}{N} \mathopen{}\left( 1 - \frac{N_{3}}{N} \right)\mathclose{} \\
    &\qquad\qquad + 18 \frac{N_{0}}{N} \frac{N_{1}}{N}
    - 32 \frac{N_{0}}{N} \frac{N_{2}}{N}
    + 16 \frac{N_{0}}{N} \frac{N_{3}}{N}
    + 288 \frac{N_{1}}{N} \frac{N_{2}}{N}
    - 144 \frac{N_{1}}{N} \frac{N_{3}}{N}
    + 256 \frac{N_{2}}{N} \frac{N_{3}}{N}
    \biggr]\mathclose{} .
\end{split}
\end{equation}
Again estimating $N_{k} / N$ by $K_{k} / K$, we now have an estimator for $\Var(\hat{S})$ given by
\begin{equation}
\begin{split}
    \frac{1}{\hbar^{2}} \hat{\sigma}(\hat{S})^{2}
    & = \frac{8}{3} \mathopen{}\left( \frac{l}{l_{p}} \right)^{4} \frac{N^{2}}{K} \frac{N - K}{N - 1} \mathopen{}\biggl[
    \frac{K_{0}}{K} \mathopen{}\left( 1 - \frac{K_{0}}{K} \right)\mathclose{}
    + 81 \frac{K_{1}}{K} \mathopen{}\left( 1 - \frac{K_{1}}{K} \right)\mathclose{}
    + 256 \frac{K_{2}}{K} \mathopen{}\left( 1 - \frac{K_{2}}{K} \right)\mathclose{}
    + 64 \frac{K_{3}}{K} \mathopen{}\left( 1 - \frac{K_{3}}{K} \right)\mathclose{} \\
    &\qquad + 18 \frac{K_{0}}{K} \frac{K_{1}}{K}
    - 32 \frac{K_{0}}{K} \frac{K_{2}}{K}
    + 16 \frac{K_{0}}{K} \frac{K_{3}}{K}
    + 288 \frac{K_{1}}{K} \frac{K_{2}}{K}
    - 144 \frac{K_{1}}{K} \frac{K_{3}}{K}
    + 256 \frac{K_{2}}{K} \frac{K_{3}}{K}
    \biggr]\mathclose{} .
\end{split}
\end{equation}
The square root of this gives an estimator for the standard error of the action
\begin{equation}
\label{eq:random_se_full}
\begin{split}
    \frac{1}{\hbar} \hat{\sigma}(\hat{S})
    & = \frac{4}{\sqrt{6}} \mathopen{}\left( \frac{l}{l_{p}} \right)^{2} \frac{N}{\sqrt{K}} \sqrt{\frac{N-K}{N-1}} \mathopen{}\biggl[
    \frac{K_{0}}{K} \mathopen{}\left( 1 - \frac{K_{0}}{K} \right)\mathclose{}
    + 81 \frac{K_{1}}{K} \mathopen{}\left( 1 - \frac{K_{1}}{K} \right)\mathclose{}
    + 256 \frac{K_{2}}{K} \mathopen{}\left( 1 - \frac{K_{2}}{K} \right)\mathclose{}
    + 64 \frac{K_{3}}{K} \mathopen{}\left( 1 - \frac{K_{3}}{K} \right)\mathclose{} \\
    &\qquad + 18 \frac{K_{0}}{K} \frac{K_{1}}{K}
    - 32 \frac{K_{0}}{K} \frac{K_{2}}{K}
    + 16 \frac{K_{0}}{K} \frac{K_{3}}{K}
    + 288 \frac{K_{1}}{K} \frac{K_{2}}{K}
    - 144 \frac{K_{1}}{K} \frac{K_{3}}{K}
    + 256 \frac{K_{2}}{K} \frac{K_{3}}{K}
    \biggr]^{\frac{1}{2}} .
\end{split}
\end{equation}

We may also exhibit a simpler bound on the standard error.
Due to the subadditivity of the standard deviation,
\begin{equation}
\begin{split}
    \frac{1}{\hbar} \sigma(\hat{S})
    & \leq \frac{4}{\sqrt{6}} \mathopen{}\left( \frac{l}{l_{p}} \right)^{2} \frac{N}{K} [\sigma(K_{0}) + 9 \sigma(K_{1}) + 16 \sigma(K_{2}) + 8 \sigma(K_{3})] \\
    & = \frac{4}{\sqrt{6}} \mathopen{}\left( \frac{l}{l_{p}} \right)^{2} \frac{N}{\sqrt{K}} \sqrt{\frac{N-K}{N-1}}
    \begin{aligned}[t]
        \mathopen{}\Biggl[ &
        \sqrt{\frac{N_{0}}{N} \mathopen{}\left( 1 - \frac{N_{0}}{N} \right)}
        + 9 \sqrt{\frac{N_{1}}{N} \mathopen{}\left( 1 - \frac{N_{1}}{N} \right)} \\
        & + 16 \sqrt{\frac{N_{2}}{N} \mathopen{}\left( 1 - \frac{N_{2}}{N} \right)}
        + 8 \sqrt{\frac{N_{3}}{N} \mathopen{}\left( 1 - \frac{N_{3}}{N} \right)}
        \Biggr]\mathclose{} .
    \end{aligned}
\end{split}
\end{equation}
Estimating $N_{k} / N$ by $K_{k} / K$ as usual gives an estimate for the standard error
\begin{equation}
\label{eq:random_se_subadditive_bound}
    \frac{1}{\hbar} \hat{\sigma}(\hat{S})
    \leq \frac{4}{\sqrt{6}} \mathopen{}\left( \frac{l}{l_{p}} \right)^{2} \frac{N}{\sqrt{K}} \sqrt{\frac{N-K}{N-1}}
    \begin{aligned}[t]
        \mathopen{}\Biggl[ &
        \sqrt{\frac{K_{0}}{K} \mathopen{}\left( 1 - \frac{K_{0}}{K} \right)}
        + 9 \sqrt{\frac{K_{1}}{K} \mathopen{}\left( 1 - \frac{K_{1}}{K} \right)} \\
        & + 16 \sqrt{\frac{K_{2}}{K} \mathopen{}\left( 1 - \frac{K_{2}}{K} \right)}
        + 8 \sqrt{\frac{K_{3}}{K} \mathopen{}\left( 1 - \frac{K_{3}}{K} \right)}
        \Biggr]\mathclose{} .
    \end{aligned}
\end{equation}
We also immediately have the bound
\begin{equation}
    \frac{1}{\hbar} \sigma(\hat{S})
    \leq \frac{68}{\sqrt{6}} \mathopen{}\left( \frac{l}{l_{p}} \right)^{2} \frac{N}{\sqrt{K}} \sqrt{\frac{N-K}{N-1}} .
\end{equation}

\printbibliography

\end{document}